\documentclass[11pt,a4paper]{article}

\usepackage{a4wide}

\usepackage[english]{babel}
\usepackage[utf8]{inputenc}
\usepackage[T1]{fontenc}

\usepackage{cite}               
\usepackage{bbm}
\usepackage[affil-it]{authblk}

\usepackage{dsfont}
\usepackage{braket,mleftright}
\usepackage{upgreek}
\usepackage{etoolbox}

\usepackage{pgfplots, graphicx, tikz}
\usetikzlibrary{arrows}
\usetikzlibrary{decorations.markings}
\usepackage{subcaption}
\usepackage[labelfont=bf]{caption}
\usepackage{algorithm}
\usepackage{algpseudocode}

\usepackage{amsmath, amssymb, amsfonts, amssymb, commath}
\usepackage{physics}

\usepackage[colorinlistoftodos]{todonotes}
\usepackage{xcolor}

\newcommand{\Rn}{\mathds{R}^{d}}
\newcommand{\R}{\mathds{R}}

\newcommand{\Lq}[1]{L^{#1}}
\newcommand{\inner}[2]{\langle #1,#2 \rangle}
\newcommand{\vect}[1]{\mathbf{#1}}
\newcommand{\Hnorm}[1]{||#1||_{H_{0}^{1}}}
\newcommand{\subnorm}[2]{||#1||_{#2}}
\newcommand{\expct}{\mathop{\mathbb{E}}}
\newcommand{\Rnd}{(\mathds{R}^{d})^{N}}
\newcommand{\Rmd}{(\mathds{R}^{d})^{N-1}}

\usepackage[thmmarks,  thref, amsmath]{ntheorem}

\theoremstyle{plain}

\newtheorem{theorem}{Theorem}
\newtheorem{model}{Model}

\theorembodyfont{\upshape}

\theoremheaderfont{\itshape}
\theorembodyfont{\upshape}

\theoremstyle{break}
\theoremheaderfont{\scshape\bfseries}
\theorembodyfont{\upshape}
\theoremsymbol{}

\theoremstyle{nonumberbreak}
\theoremheaderfont{\scshape\bfseries}
\theorembodyfont{\upshape}
\theoremsymbol{\rule{0.7em}{0.7em}}%
\theorempostwork{\setcounter{case}{0}\setcounter{ppart}{0}}

\newtheorem{proof}{Proof}

\theoremstyle{nonumberplain}
\theoremheaderfont{\itshape\bfseries}
\theorembodyfont{\upshape}
\theoremsymbol{}

\theoremstyle{empty}
\theoremsymbol{}

\begin{document}

\author{Simone G\"ottlich \thanks{Electronic address: \texttt{goettlich@uni-mannheim.de}}}
\author{Steven Knapp \thanks{Electronic address: \texttt{stknapp@mail.uni-mannheim.de}}}
\affil{University of Mannheim, Department of Mathematics, 68159 Mannheim, Germany}
\author{Dylan Weber \thanks{Electronic address: \texttt{djweber3@asu.edu}}}
\affil{Arizona State University,  School of Mathematical and Statistical Sciences, Tempe, AZ 85257-1804, USA}

\date{Dated: \today}

\title{The food seeking behavior of slime mold: a macroscopic approach}

%
%
\maketitle

\begin{abstract}
\noindent
Starting from a particle model we derive a macroscopic aggregation-diffusion equation
for the evolution of slime mold under the assumption of propagation of chaos in the large particle limit.
We analyze properties of the macroscopic model in the stationary case and study the behavior of the slime
mold between food sources. The efficient numerical simulation of the aggregation-diffusion equation
allows for a detailed analysis of the interplay between the different regimes drift, interaction and diffusion.
\end{abstract}





\tableofcontents

\section{Introduction and Background}
Physarum Polycephalum, or a slime mold, is an amoeboid organism that is notable for its ability to perform complex tasks despite its relatively simple biological strucutre and lack of a brain.    Surprisingly, slime molds are known to be able to solve mazes, construct robust networks, and solve shortest path and spanning tree problems \cite{bonifaci_physarum_2012, nakagaki_t_obtaining_2004, adamatzky_slime_2012, adamatzky_reaction-diffusion_2009, latty_speedaccuracy_2011, baumgarten_plasmodial_2010,adamatzky_developing_2009, nakagaki_smart_2001,oettmeier_physarum_2017}.  The simplicity of the slime mold organism and the complexity of its emergent behaviors have inspired research into Physarum Polycephalum as a model for collective behavior \cite{jones_approximating_2009, jones_influences_2011, tero_physarum_2006, jones_characteristics_2010 }.  A hallmark of research in collective dynamics is complex global structure emerging as a result of relatively simple local interactions between agents.  Indeed, cellular automata models inspired by physarum polycephalum food seeking behavior have been found to be good models of other collective phenomena such as the formation of transportation networks in a country or the "cosmic web" of stellar material between galaxies \cite{burchett_revealing_2020, adamatzky_brazilian_2011}.

When seeking food, physarum polycephalum moves to concentrate most of its mass on food sources - however it maintains a network of mass between food sources.  In the absence of food it begins to retract to a central mass.  Past research into models of physarum polycephalum mainly attempt to model this food-seeking behavior and have taken several different directions.  Many are discrete cellular automata models - the slime mold is represented by a collection of agents that evolve discretely in time and space.  Different interaction rules and sometimes even different "types" of agents are employed to model different slime mold behaviors however almost all models include a rule that causes slime mold agents to move towards food sources \cite{liu_physarum_2013, wu_new_2015, liu_new_2017, tsompanas_modeling_2012, gunji_minimal_2008}.  Other approaches note that  physarum polycephalum moves and forms networks by moving fluid through tubules within its mass - much like an electrical network \cite{tero_physarum_2006, tero_flow-network_2008, tero_mathematical_2007}.  These approaches focus on modeling physarum polycephalum as a network of ordinary differential equations and attempt to model the network that physarum polycephalum builds between food sources once they are found.  In particular they do not attempt to model the food seeking behavior of the organism.

In this manuscript we propose a model hierarchy aimed at modeling the food seeking behavior of Physarum Polycephalum.  The main behaviors we attempt to model are aggregation of mass on food sources while maintaining connected paths between food sources and retraction in the absence of food sources.  We first define an agent based model that is continuous in time and space through a system of stochastic differential equations.  The evolution of each agent is governed by three terms - an interaction term with other agents that causes agents to aggregate, a drift term that causes agents to move towards food sources and a noise term intended to model slime mold foraging behavior in the absence of food sources.  We then rigorously show under the assumption of propagation of chaos that as the number of agents approaches infinity, the agent based model converges in a sense to a deterministic density based model in the form of an aggregation-diffusion equation.  This macroscopic model includes terms that correspond to the drift, interaction and noise terms present in the microscopic model.  We then embark on an analysis of the macroscopic equation from the point of view of modeling Physarum food seeking behavior.  We first discuss a framework for choosing an appropriate interaction kernel through an analysis of the steady states of the macroscopic model in the case that there are no food sources.  Then, we rigorously show that there is a parameter regime in which the only possible steady state of the macroscopic model is zero and thus this regime is not suitable for modeling Physarum food seeking behavior.  Finally, we simulate the macroscopic model using the \textit{blob method for aggregation-diffusion equations} presented in \cite{carrillo_blob_2019} and find that it is necessary to scale the three main terms of the model in order to reproduce aggregation on food sources.  We then simulate the evolution of the model in different scaling regimes and with different choices of interaction kernel and qualitatively examine which regimes result in the main behaviors we attempt to model.  We find that there are several regimes that result in the qualitative desired behaviors.

\section{Model Derivation}
The main features of slime mold food seeking behavior are aggregation of mass on food sources while maintaining a connected mass and retraction in the absence of food sources.  To model these behaviors we first introduce an agent-based model. Here, the slime mold is represented by a swarm of $N$ agents who's trajectories, $\vect{X}_{i}(t) \in \Rn$, are modeled as a system of stochastic differential equations.  In the following we will refer to this model as the \textit{microscopic model} in order to reinforce the intuition that the agents represent slime mold "particles".
\begin{model}[Microscopic Model]
Consider a collection of $N$ agents with (random) positions \\ $(\vect{X}_{1}(t),...,\vect{X}_{n}(t))\in \Rnd$.  The \textbf{microscopic model} is given by the evolution:
\begin{equation}
  \begin{aligned}
  \label{eq:microscopic_model}
  d\vect{X}_{i}(t) = A\nabla V(\vect{X}_{i}(t))dt + B \frac{1}{N -1}\sum_{j\neq i}^{N}\nabla W(\vect{X}_{j}(t) - \vect{X}_{i}(t))dt + C d\vect{B}_{i}(t)
\end{aligned}
\end{equation}
where $\set{\vect{B}_{i}}_{i}$ are independent, standard $\Rn$-valued Brownian motions, $V,W:\Rn\rightarrow\R$ are $C^1$ functions,
$\set{\vect{X}_{i}(0)}_{i}$ are independent and identically distributed random vairables independent of $\set{\vect{B}_{i}}$ with probability density function given by  $\rho_{0}:\Rn\rightarrow \R$, and $A,B,C>0$.
\end{model}
We note that if we assume that $\nabla V$ and $\nabla W$ are globally Lipschitz and further if there exists $K_{1}, K_{2} >0$ such that for $x \in \Rn$:
\begin{equation*}
  |\nabla V(x)| \leq K_{1}(1 + |x|)\quad\text{and} |\nabla W(x)| \leq K_{2}(1 + |x|),
\end{equation*}
then it can be shown that \eqref{eq:microscopic_model} has a unique strong solution.

Intuitively, $V$ represents the density of chemo-attractants given off by food sources.  For this reason, we usually think of $V$ as a sum of radially symmetric, positive functions. Each term represents the attractants emanating from a single food source.  A prototypical example would be a sum of Gaussians centered at different positions.  $W$ represents the interaction between agents.  Notice that, roughly, if $\nabla W(\vect{X}_{j}(t) - \vect{X}_{i}(t)) \approx \vect{X}_{j}(t) - \vect{X}_{i}(t)$, then the interaction causes agents to \textit{attract} each other. lLikewise if $\nabla W(\vect{X}_{j}(t) - \vect{X}_{i}(t)) \approx \vect{X}_{i}(t) - \vect{X}_{j}(t)$, then the agents will \textit{repel} each other.  As we have noted, Physarum tends to maintain a connected mass, even when aggregating around disparate food sources.  Therefore, this interaction term should cause agents to remain locally close, i.e. agents should not repel each other at long ranges.  The noise term, $Cd\vect{B}_{i}(t)$, models physarum foraging behavior.  In the absence of chemo-attractants, Physarum forages by diffusing outward from its initial position while still maintaining a connected mass.  Here, in the absence of a food source ($V = 0$) the noise term should cause particles to spread from their initial positions while the interaction term causes them to still have some propensity to remain locally "together".

  We include the parameters $A,B$ and $C$ for the case that simulations of \eqref{eq:microscopic_model} demonstrate the need to scale the terms' influence on each particles trajectory to produce evolutions that resemble Physarum.  However, any simulation of \eqref{eq:microscopic_model} that would give an accurate picture of how \eqref{eq:microscopic_model} models Physarum would have to use a very large number of particles; this becomes computationally intractable.  Therefore, we must find a way to examine the behavior or \eqref{eq:microscopic_model}, for a large number of particles, without computing the trajectory of every particle explicitly.  We will show that in the large particle limit that the evolution of the marginal distribution of any particle is given by an \textit{aggregation-diffusion equation}.  There exist computationally efficient methods for simulating such equations, therefore we can examine the behavior of \eqref{eq:microscopic_model} by simulating the evolution of the marginal density.  The technique we will employ to derive an appropriate model is the assumption of \textit{propagation of chaos}; i.e. if the number of agents agents is large, the trajectories of any two agents can be assumed to be independent.  This is true in many cases, see for example \cite{chen_modeling_2018, chen_mean_2017}.

\begin{theorem}\label{thm:convergence}
In the $N\rightarrow \infty$ limit and for a given, $A,B,C > 0$, $W,V \in C^1(\Rn)$ and $\rho_{0}$ the evolution of the marginal distribution of any agent evolving according to \eqref{eq:microscopic_model} is a solution of the aggregation-diffusion equation
\begin{align*}
  \partial_{t}\rho(x, t)&= A\partial_{x}\Big[\nabla{V}(x)\rho(x, t)\Big] + B\partial_{x}\Big[ \rho(x,t)\nabla W * \rho(x,t)\Big] + C\Delta \rho(x, t),\\
  \rho(x, 0) &= \rho_{0}(x)\;,\;x\in\Rn
\end{align*}
under the assumption of propagation of chaos.
\end{theorem}

\begin{proof}
Consider a finite group of $N$ agents, $(X_{i}(t), t\geq 0){i}^{N}\subseteq \Rn, i=1,...,N$ evolving according to \eqref{eq:microscopic_model}.  Let $\mu^{N}$ be the empirical distribution of the configuration of agents.  That is:
\begin{align*}
  \mu^{N}(t) = \frac{1}{N}\sum_{i}\delta_{X_{i}(t)}
\end{align*}
where $\delta_{X_{i}(t)}$ is the Dirac measure with unit mass at $X_{i}(t)$.  Let $F$ be a test function in $C_{c}^{\infty}(\Rnd)$, therefore by definition of the Dirac measure we deduce:
\begin{align*}
   \inner{\rho^{N}}{F} = \frac{1}{N}\sum_{i}F(X_{i}(t))
\end{align*}
we first derive an expression for the time evolution of the expectation of this quantity. Let $\vect{X}(t) = (X_{1}(t),...,X_{N}(t))$.  By conditioning we find
\begin{equation}
  \label{eq:condition}
  \begin{aligned}
  \frac{\expct[F(\vect{X}(t + \Delta t)] - \expct[F(\vect{X}(t))]}{\Delta t} = \int_{\Rnd}\frac{\expct[F(\vect{X}(t + \Delta t)| \vect{X(t)}= \vect{x}] - F(\vect{x})}{\Delta t}\rho^{N}(\vect{x}, t)d\vect{x}
  \end{aligned}
\end{equation}
where $\rho^{N}(\vect{x}, t)$ is the joint distribution of the collection of agents.  Using Ito's Lemma (see e.g. \cite{oksendal_stochastic_2003}) we can further expand the conditioned term in the above
\begin{equation}
  \label{eq:ito}
\begin{aligned}
  \expct\Big[F(\vect{X}(t + \Delta t))| \vect{X(t)} = \vect{x}\Big] &= F(\vect{x})\\
                                                           &+\expct\Big[\sum_{j = 1}^{N}\int_{t}^{t+\Delta t}\partial_{x_{j}}[F(\vect{X}(t))]dX_{j}(s)|\vect{X}(t) = \vect{x}\Big]\\
                                                           &+ \expct\Big[\frac{1}{2}\sum_{j,k}^{N}\int_{t}^{t+\Delta t}\partial_{x_{j}}\partial_{x_{k}}[F(\vect{X}(t))]d[X_{i}(s),X_{j}(s)]|\vect{X}(t) = \vect{x}\Big].
\end{aligned}
\end{equation}
We first focus on the third term of \eqref{eq:ito}.  Notice that as each $X_{i}$ evolves according to \eqref{eq:microscopic_model}, since the noise terms for disparate agents are independent, we have that:
\begin{align*}
  d[X_{i}(s), X_{j}(s)] = \begin{cases}
                              0\quad&\text{if}\quad i\neq j \\
                              Cds\quad&\text{if}\quad i = j \\
                          \end{cases}.
\end{align*}
Therefore, the third term of \eqref{eq:ito} simplifies as:
\begin{equation}
  \label{eq:diffuse}
  \expct\Big[\frac{1}{2}\sum_{j,k}^{N}\int_{t}^{t+\Delta t}\partial_{x_{j}}\partial_{x_{k}}[F(\vect{X}(t))]d[X_{i}(s),X_{j}(s)]|\vect{X}(t) = \vect{x}\Big] = \expct\Big[ C\int_{t}^{t + \Delta t} \Delta F(\vect{X}(s))ds|\vect{X}(t) = \vect{x}\Big].
\end{equation}
Now, turning to the second term of \eqref{eq:ito} we find by appyling \eqref{eq:microscopic_model} that:
\begin{equation}
  \label{eq:advect}
  \begin{aligned}
  \expct\Big[\sum_{j = 1}^{N}\int_{t}^{t+\Delta t}\partial_{x_{j}}[&F(\vect{X}(t))]dX_{j}(s)|\vect{X}(t) = \vect{x} \Big] = \expct\Big[\sum_{j=1}^{N}\int_{t}^{t+\Delta t}\partial_{x_{j}}[F(\vect{X}(s))]\Big(A\nabla V(X_{j}(s))\Big)ds|\vect{X}(t) = \vect{x}\Big]\\
  &+\expct\Big[\sum_{j=1}^{N}\int_{t}^{t+\Delta t}\partial_{x_{j}}[F(\vect{X}(s))]\Big(\frac{B}{N-1}\sum_{k\neq j}\nabla W(X_{j}(s) - X_{k}(s))\Big)ds|\vect{X}(t) = \vect{x}\Big]\\
  &+\expct\Big[\sum_{j}\int_{t}^{t+\Delta t}\partial_{x_{j}}[F(\vect{X}(s))]\Big(CdB_{j}(s)\Big)|\vect{X}(t) = \vect{x}\Big].\\
\end{aligned}
\end{equation}
As the stochastic integral is a martingale, we have that the last term of \eqref{eq:advect} is equal to $0$.  Therefore, combining \eqref{eq:diffuse} and \eqref{eq:advect}, we have that:
\begin{equation}
  \label{eq:dt_rho}
  \begin{aligned}
      \expct\Big[F(\vect{X}(t + \Delta t))| \vect{X(t)} &= \vect{x}\Big] = F(\vect{x}) \\
        &+ \expct\Big[\sum_{j=1}^{N}\int_{t}^{t+\Delta t}\partial_{x_{j}}[F(\vect{X}(s))]\Big(A\nabla V(X_{j}(s))\Big)ds|\vect{X}(t) = \vect{x}\Big] \\
        &+\expct\Big[\sum_{j=1}^{N}\int_{t}^{t+\Delta t}\partial_{x_{j}}[F(\vect{X}(s))]\Big(\frac{B}{N-1}\sum_{k\neq j}\nabla W(X_{j}(s) - X_{k}(s))\Big)ds|\vect{X}(t) = \vect{x}\Big]\\
        &+\expct\Big[ C\int_{t}^{t + \Delta t} \Delta F(\vect{X}(s))ds|\vect{X}(t) = \vect{x}\Big].
  \end{aligned}
\end{equation}
Therefore, by \eqref{eq:condition} and the continuity of the trajectories, \eqref{eq:dt_rho} implies $P$ almost-surely that:
\begin{equation}
  \begin{aligned}
  \label{eq:rho_N_weak}
  \partial_{t}\int_{\Rnd}F(\vect{x})\rho^{N}(\vect{x}, t)d\vect{x} &= \int_{\Rnd} C\Delta F(\vect{x})\rho^{N}(\vect{x}, t)d\vect{x} \\
                                      &+ \int_{\Rnd}\sum_{j=1}^{N}\partial_{x_{j}}[F(\vect{x})] \Big(A\nabla V(X_{j}(t))\Big)\rho^{N}(\vect{x},t)d\vect{x}\\
                                      &+ \int_{\Rnd}\sum_{j=1}^{N}\partial_{x_{j}}[F(\vect{x})]\Big(\frac{B}{N-1}\sum_{k\neq j}\nabla W(X_{j}(t) - X_{k}(t))\Big)\rho^{N}(\vect{x}, t).
  \end{aligned}
\end{equation}
Therefore, $\rho^{N}$ is a weak solution to the following initial value problem:
\begin{equation}
  \label{eq:rho_N_time_evolve}
  \begin{aligned}
    \partial_{t} \rho^{N}(\vect{x}, t) &- C\Delta\rho^{N}(\vect{x}, t) -\sum_{j=1}^{N}\partial_{x_{j}}\Big[A\nabla V(x_{j})\rho^{N}(\vect{x}, t)\Big] \\
                                       &- \sum_{j=1}^{N}\partial_{x_{j}}\Big[\rho^{N}(\vect{x}, t)\frac{B}{N-1}\sum_{k\neq j}\nabla W (x_{j} - x_{k}(t))\Big] = 0, \\
                                    &\rho^{N}(\vect{x}, 0) = \prod_{j = 1}^{N}\rho_{0}(x_{j}).
  \end{aligned}
\end{equation}
Using \eqref{eq:rho_N_time_evolve}, we now compute the time evolution of the marginal distribution of one agent in the $N\rightarrow \infty$ limit under the assumption of \textit{propagation of chaos}, that is for any finite collection of $M$ agents we have that their joint distribution satisfies:
\begin{align}
  \label{eq:prop_chaos}
  \rho^{M}(x_{1},...,x_{N},t) = \prod_{i =1}^{N}\rho(x_{i}, t).
\end{align}
Without loss of generality we will consider the marginal distribution of the first agent which is given by:
\begin{align*}
  \rho(x_{1}, t) = \int_{\Rmd}\rho^{N}(x_{1},...,x_{N})d(x_{2},...,x_{N})
\end{align*}
and therefore, by \eqref{eq:rho_N_time_evolve}, we have that $\rho(x_{1}, t)$ satisfies (weakly):
\begin{equation}
  \label{eq:marginal_evolution}
  \begin{aligned}
    \partial_{t}\rho(x_{1}, t) &= \sum_{j = 1}^{N}\int_{\Rmd} \partial_{x_{j}}\Big[A\nabla V(x_{j})\rho^{N}(\vect{x}, t)\Big]d(x_{2},...,x_{N})\\
                                 &+ \sum_{j= 1}^{N}\int_{\Rmd}\partial_{x_{j}}\Big[\frac{B}{N-1}\sum_{k\neq j}\nabla W(x_{j} - x_{k})\rho^{N}(\vect{x}, t)\Big]d(x_{2},...,x_{N})\\
                                 &+ \int_{\Rmd}C\Delta \rho^{N}(\vect{x}, t)d(x_{2},...,x_{N}).\\
  \end{aligned}
\end{equation}
We analyze \eqref{eq:marginal_evolution} term by term.  First, by applying \eqref{eq:prop_chaos} we find that
\begin{equation}
  \label{eq:diffuse_term}
  \int_{\Rmd}C\Delta \rho^{N}(\vect{x}, t)d(x_{2},...,x_{N}) = C\Delta \rho(x_{1}, t).
\end{equation}
Next, if we again apply \eqref{eq:prop_chaos} and integrate by parts in each term of the sum where $j\neq1$ we find that
\begin{equation}
  \label{eq:drift_term}
  \sum_{j=1}^{N}\int_{\Rmd} \partial_{x_{j}}\Big[A\nabla V(X_{j}(t))\rho^{N}(\vect{x}, t)\Big]d(x_{2},...,x_{N}) = \partial_{x_{1}}\Big[A\nabla{V}(x_{1})\rho(x_{1}, t)\Big].
\end{equation}
Similarly we deduce
\begin{equation}
  \label{eq:interaction_term_intermediate}
  \begin{aligned}
  \sum_{j=1}^{N}&\int_{\Rmd}\partial_{x_{j}}\Big[\frac{B}{N-1}\sum_{k\neq j}^{N}\nabla W(x_{j} - x_{k})\rho^{N}(\vect{x}, t)\Big]d(x_{2},...,x_{N}) = \\
  &                              \int_{\Rmd}\partial_{x_{1}}\Big[\frac{B}{N-1}\sum_{k=2}^{N}\nabla W(x_{1} - x_{k})\rho^{N}(\vect{x}, t)\Big]d(x_{2},...,x_{N}).
  \end{aligned}
\end{equation}
We can simplify \eqref{eq:interaction_term_intermediate} further by integrating under the divergence operator and applying \eqref{eq:prop_chaos}, notice that:
\begin{equation}
  \begin{aligned}
    &\int_{\Rmd}\sum_{k=2}^{N}\nabla W(x_{1} - x_{k})\rho^{N}(\vect{x}, t)d(x_{2},...,x_{N}) \\
    &=\sum_{k=2}^{N}\int_{\Rmd}\nabla W(x_{1} - x_{k})\Big(\rho(x_{1},t)\rho(x_{2},t),...,\rho(x_{N},t)\Big)d(x_{2},...,x_{N})\\
    &= \sum_{k=2}^{N}\rho(x_{1},t)\int_{\Rn}\nabla W(x_{1} - x_{k})\rho(x_{k},t)dx_{k}\\
    &=(N-1)\rho(x_{1},t)\nabla W *\rho(x_{1}).
  \end{aligned}
\end{equation}
Therefore:
\begin{equation}
  \label{eq:interaction_term}
  \begin{aligned}
    \int_{\Rmd}\partial_{x_{1}}\Big[\frac{B}{N-1}\sum_{k=2}^{N}\nabla W(x_{1} - x_{k})\rho^{N}(\vect{x}, t)\Big]d(x_{2},...,x_{N}) = \partial_{x_{1}}\Big[B \rho(x_{1},t)\nabla W * \rho(x_{1},t)\Big].
  \end{aligned}
\end{equation}
Therefore, combining \eqref{eq:diffuse_term}, \eqref{eq:drift_term} and \eqref{eq:interaction_term} we have by \eqref{eq:marginal_evolution} that the marginal of any agent satisfies:
\begin{equation}
  \label{eq:macro_model}
  \begin{aligned}
  \partial_{t}\rho(x, t)&= A\partial_{x}\Big[\nabla{V}(x)\rho(x, t)\Big] + B\partial_{x}\Big[ \rho(x,t)\nabla W * \rho(x,t)\Big] + C\Delta \rho(x, t)\\
  \rho(x, 0) &= \rho_{0}(x)
  \end{aligned}
\end{equation}
as desired.
\end{proof}

As the equation derived in Theorem \ref{thm:convergence} describes how the entire collection of particles evolves from a density standpoint, in the following we will refer to it as the \textit{macroscopic model}.

\begin{model}[Macroscopic model]
  The \textbf{macroscopic model} is given by the evolution of the aggregation-diffusion equation
  \begin{equation}
  \begin{aligned}\label{eq:macroscopic_model}
    \rho_{t}(t,x) &= A(\nabla \cdot (\nabla V\rho))  + B(\nabla \cdot((\nabla W * \rho)\rho)) + C(\Delta\rho)\quad x\in \Omega \subseteq \Rn, \\
    \rho(0, x) &= \rho_{0}
  \end{aligned}.
\end{equation}
  for $W,V \in C^1(\Rn)$ and $\rho_{0}\in H^1(\Rn)$ with $\int_{\Omega} \rho_{0}(x)dx=1.$
\end{model}

Here, $V$ and $W$ continue to represent the chemoattractant of food sources and the slime mold's propensity to aggregate respectively.  In the following we will study \eqref{eq:macroscopic_model} from the viewpoint of using it to model Physarum food-seeking behavior.


\section{Model Properties}
Many questions remain about how to utilize the macroscopic equation derived in the previous section to model a slime mold.  In this section we present some analysis of stationary states of the equation in order to inform some of these modeling choices.  First, we investigate stationary states in the case where there is no food present as a way to gain some information about what choice of kernel function might be reasonable.  Inspired by \cite{nakagaki_obtaining_2004}, we believe that Gaussians represent a reasonable model for stationary states of a slime mold - however we find that Gaussian stationary states are only possible in the case that the kernel function, $W$, is quadratic.  Next, we investigate stationary states of the scaled equation.  We find through a fixed point argument that if the diffusion scaling parameter, $C$, is sufficiently higher than $A$ and $B$ that the only stationary state of the equation is $0$ (even in the presence of a food source).  This suggests that when tuning the scaling parameters that there cannot be "too much" diffusion.

\subsection{Zero-food stationary states}
In the case that there are no food sources a slime mold will "retract" to a more compact configuration \cite{nakagaki_t_obtaining_2004, nakagaki_maze-solving_2000, nakagaki_path_2001}.  Often, this configuration is roughly radially symmetric.  Given that, a reasonable model for the configuration of the slime mold in the case that there is no food could be a Gaussian.  We now investigate if there are conditions on the kernel function that are imposed by the assumption that stationary states are Gaussian.  For simplicity we work in one dimension however the calculations are analagous in higher dimensions.  In the case that there are no food sources \eqref{eq:macroscopic_model} becomes:
\begin{align}\label{eq:zero_food}
\rho_{t} = A\Delta \rho + B\nabla \cdot ((\nabla W * \rho)\rho).
\end{align}
Therefore a stationary state satisfies (in the following we will ignore the scaling parameters, $A$ and $B$, as they do not change the computation):
\begin{align*}
    -\Delta \rho = \nabla \cdot ((\nabla W *\rho)\rho)
\end{align*}
which implies that
\begin{align*}
    -\nabla \rho = (\nabla W * \rho)\rho
\end{align*}
and finally that:
\begin{align}\label{eq:0food_condition}
    \nabla W * \rho = \nabla (\text{ln}(\rho)).
\end{align}
\eqref{eq:0food_condition} represents a general condition that zero food stationary states must satisfy.  If we assume that stationary states are Gaussian, i.e. that:
\begin{align*}
    \rho(x) = \beta \exp(-x^{2}\tau)\;\beta,\tau > 0
\end{align*}
then by \eqref{eq:0food_condition} we have that:
\begin{align*}
    \frac{\dif}{\dif x} \Big [ \int_{-\infty}^{\infty} W(x-y)\beta\exp(-y^{2}\tau)\dif y \Big] &= - \frac{\dif}{\dif x}\Big[ \text{ln}(\beta \exp(-x^{2}\tau)\Big] \\
    \implies \int_{-\infty}^{\infty} W'(y)\beta\exp(-(y-x)^{2}\tau)\dif y &= 2x\tau. \\
\end{align*}
The above can only hold if $W'(y) = ay$ for some $a\in \R$.  Therefore, Gaussians are zero food stationary states only if the interaction kernel is quadratic.  Further, by plugging in $ay$ for $W'(y)$ in the above we can see that the following must hold: \begin{align*}
    \beta a = \frac{2\tau^{\frac{3}{2}}}{\sqrt{\pi}}
\end{align*}
So, for a given quadratic kernel there are a family of Gaussians that are potential zero food stationary states (depending on the mass of the initial profile).  For an illustration of this fact see Figure \ref{fig:0food_retract}.  Here, we simulate \eqref{eq:macroscopic_model} in the case of no drift term using the "blob method" for aggregation-diffusion equations introduced in \cite{carrillo_blob_2019}.  We will discuss the numerical method in more detail in the Numerics section of the paper.  We choose an initial condition representative of a slime mold agregated around two food sources at $x = 1$ and $x = -1$.  We find in accordance with our calculation above that the mass profile appears to converge to a Gaussian configuration.  This also illustrates a modeling property of \eqref{eq:macroscopic_model}; in the case that food sources "run out" the profile will "retract" to a central configuration.

Effectively, by starting from information about the zero food stationary state, we have "solved" for the correct kernel.  This analysis suggests an empirical method for determining the correct kernel function from experimental data.  Instead of our assumption that zero food stationary states are Gaussian, statistical analysis of slime mold configurations in the presence of zero food sources could provide information about the stationary state configuration.  This information could then be used to "solve" for the kernel that results in this "correct" stationary state.

\begin{figure}
  \centering
  \includegraphics[scale = 0.6]{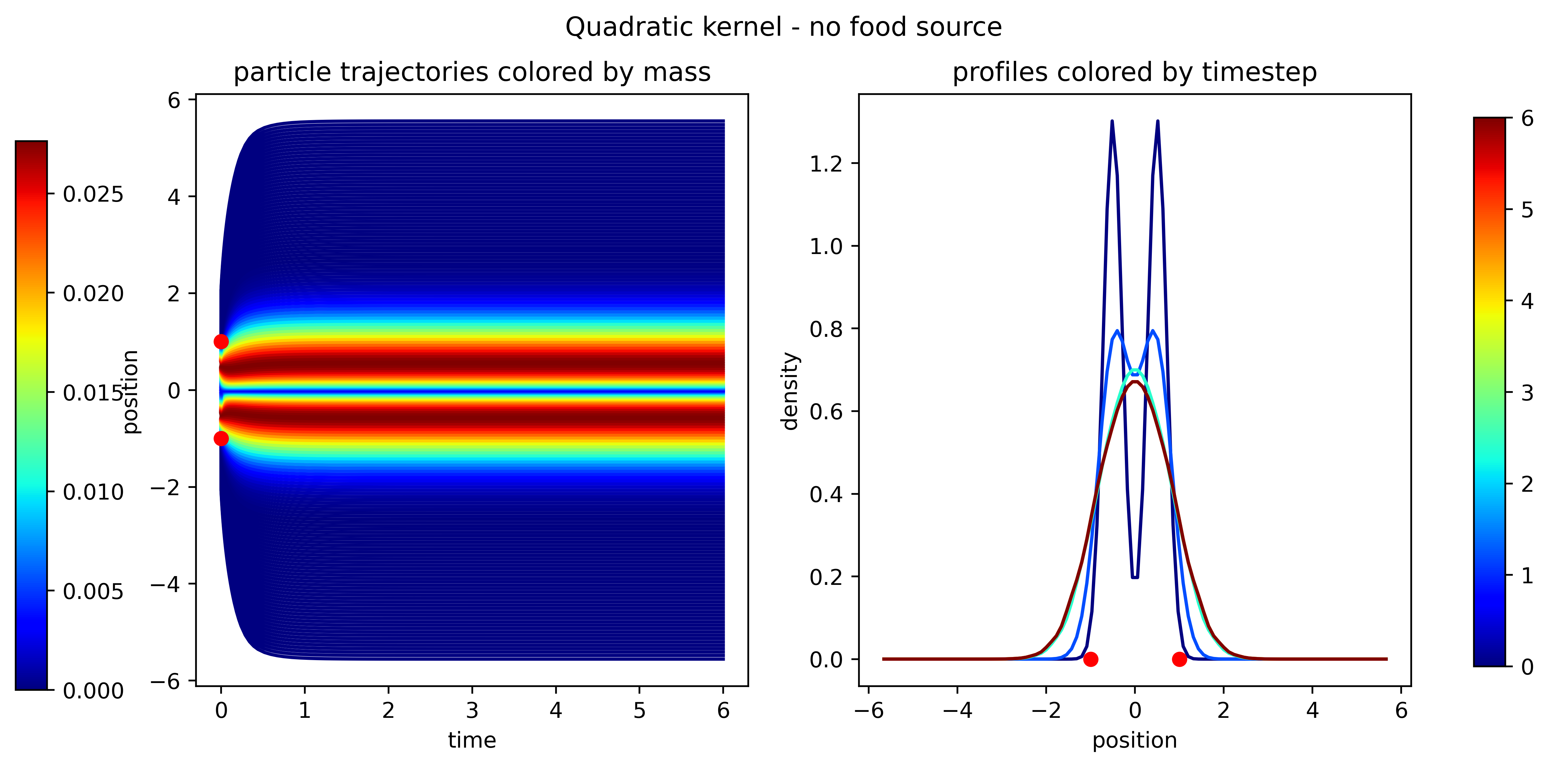}
  \caption{The evolution of \eqref{eq:macroscopic_model} in the case of a quadratic kernel with no drift term.  The initial condition is chosen to represent the state of a slime mold after it has agregated around two food sources at $x =1$ and $x = -1$.  Notice that the profiles appear to converge on a central Gaussian configuration - an illustration of the fact that Gaussians are stationary states of the equation in the case of a quadratic kernel. }
  \label{fig:0food_retract}
\end{figure}

\subsection{High diffusion stationary states}
We now turn to considering stationary states of the scaled equation:
\begin{align}
\rho_{t} = A(\nabla \cdot (\nabla V\rho))  + B(\nabla \cdot((\nabla W * \rho)\rho)) + C(\Delta\rho).
\end{align}
We will see in a later section that the "fair" regime $(A=B=C)$ does not result in food-source aggregation.  Therefore we are motivated to study different scalings of the terms in \eqref{eq:macroscopic_model}.  Here, we will rule out a large regime of scaling parameters by employing a fixed point argument to show that if the diffusion parameter, $C$, is sufficiently larger then the drift and interaction parameters then the only stationary state of \eqref{eq:macroscopic_model} is $0$ on bounded domains.
\begin{theorem}
  \label{thm:diffusion_stationary_states}
  Let $\Omega\subseteq \Rn$ be compact and connected.  For a given $A,B < 0$ and $V,W\ in H^{2}(\Omega)\cap L^{\infty}(\Omega)$ then for any sufficiently large $C > 0$, $0$ is the unique stationary state of:
  \begin{align*}
    \rho_{t}(t,x) &= A(\nabla \cdot (\nabla V\rho))  + B(\nabla \cdot((\nabla W * \rho)\rho)) + C(\Delta\rho^{m}),\quad x\in \Omega \subseteq \Rn, \\
    \rho(0, x) &= \rho_{0}(x)
  \end{align*}
  in $H_{0}^{1}(\Omega)$,
\end{theorem}
\begin{proof}
Recall that $\rho$ is a stationary state of \eqref{eq:macroscopic_model}  if:
\begin{align*}
  -\Delta \rho = \frac{A}{C}[\nabla V \cdot \nabla \rho + \Delta V\rho] + \frac{B}{C}[\nabla(W * \rho)\cdot \nabla\rho + \Delta(W * \rho)\rho].
\end{align*}
  We define:
\begin{align*}
  G(\rho) := \frac{A}{C}[\nabla V \cdot \nabla \rho + \Delta V\rho] + \frac{B}{C}[\nabla(W * \rho)\cdot \nabla\rho + \Delta(W * \rho)\rho].
\end{align*}
Given $\rho \in H_{0}^{1}(\Omega)$ we know that $G(\rho) \in L^{2}(\Omega)$ and that the equation:
\begin{align}\label{eq:stationary_equation}
  \Delta\tilde{\rho} = G(\rho)
\end{align}
has a unique solution in $H_{0}^{1}(\Omega)$.  Given $K > 0$, define $H_{0}^{1}(\Omega, K) := \set{\rho \in H_{0}^{1}(\Omega) : ||\rho||_{H_{0}^{1}} \leq K}$.  It can be shown (see, for example, Theorem 4 in section 6.3 of \cite{evans_partial_2010}) that if $\tilde{\rho}$ is the unique solution of \eqref{eq:stationary_equation} then the following estimate holds for some $D > 0$:
\begin{equation}
  \label{eq:evans}
    ||\tilde{\rho}||_{H^{2}} \leq D||G(\rho)||_{L^{2}}
\end{equation}
Therefore, by applying standard estimates we can see that for a sufficiently large choice of $C$ if $\rho \in H_{0}^{1}(\Omega, K)$ then we must also have that the solution to $\eqref{eq:stationary_equation}$, $\tilde{\rho}$, belongs to $\rho \in H_{0}^{1}(\Omega, K)$ as well.

    Define $\tilde{A}:H_{0}^{1}(\Omega, K)\rightarrow H_{0}^{1}(\Omega,K)$ via:
\begin{align*}
    \tilde{A}[\rho] = \tilde{\rho}\quad\text{where $\tilde{\rho}$ solves \eqref{eq:stationary_equation}}.
\end{align*}
We will show that $\tilde{A}$ is a contraction and thus has a unique fixed point by Banach's fixed point theorem.  Since $0$ is trivially a fixed point of $\tilde{A}$ and $K$ is arbitrary this implies that $0$ is the only fixed point of $\tilde{A}$ and thus the only stationary state of \eqref{eq:macroscopic_model} .  Let $\tilde{\rho}_{1} = \tilde{A}[\rho_{1}]$ and $\tilde{\rho}_{2} = \tilde{A}[\rho_{2}]$.  We will show that there exists a $0<\gamma(A,B,C)<1$ such that:
\begin{align*}
    \Hnorm{A[\rho_{1}] - A[\rho_{2}]} \leq \gamma(A,B,C)\Hnorm{\rho_{1} - \rho_{2}}.
\end{align*}
Since $\Omega$ is bounded by assumption we have by Poincare's inequality that:
\begin{align}\label{eq:poincare}
  \Hnorm{\tilde{\rho_{1}} - \tilde{\rho_{2}}} \leq \alpha \subnorm{D(\tilde{\rho_{1}} - \tilde{\rho_{2}})}{\Lq{2}}
\end{align}
for some $\alpha > 0$.  Therefore since $\tilde{\rho_{1}}$ and $\tilde{\rho_{2}}$ are both solutions of \eqref{eq:stationary_equation} we must have by another application of the Poincare inequality that:
\begin{align*}
  \subnorm{D(\tilde{\rho}_{1} - \tilde{\rho_{2}})}{\Lq{2}}^{2} &\leq |\int_{\Omega}D(\tilde{\rho}_{1} - \tilde{\rho_{2}}) \cdot D(\tilde{\rho}_{1} - \tilde{\rho_{2}}) dx| \\
  &= |-\int_{\Omega}(G(\rho_{1}) - G(\rho_{2}))(\tilde{\rho_{1}} - \tilde{\rho_{2}})dx| \\    &\leq \subnorm{G(\rho_{1}) - G(\rho_{2})}{\Lq{2}}\subnorm{\tilde{\rho}_{1} - \tilde{\rho_{2}}}{\Lq{2}} \\
  &\leq \alpha \subnorm{G(\rho_{1}) - G(\rho_{2})}{\Lq{2}}\subnorm{D(\tilde{\rho}_{1} - \tilde{\rho_{2}})}{\Lq{2}}.
\end{align*}
Therefore we have that:
\begin{align*}
  \subnorm{D(\tilde{\rho}_{1} - \tilde{\rho_{2}})}{\Lq{2}} \leq \alpha \subnorm{G(\rho_{1}) - G(\rho_{2})}{\Lq{2}}
\end{align*}
and by \eqref{eq:poincare} that:
\begin{align}\label{eq:G_condition}
  \Hnorm{\tilde{\rho_{1}} - \tilde{\rho_{2}}} \leq \alpha\subnorm{G(\rho_{1}) - G(\rho_{2})}{\Lq{2}}.
\end{align}
We claim that
\begin{align}\label{eq:contract}
  \subnorm{G(\rho_{1}) - G(\rho_{2})}{\Lq{2}} \leq \gamma(A,B,C)\subnorm{\rho_{1} - \rho_{2}}{\Lq{2}}.
\end{align}
for some $\gamma > 0$ that depends only on $A,B$ and $C$ for a given $W$ and $V$.
Expanding we find that:
\begin{equation}\label{eq:expanded}
  \begin{aligned}
    \subnorm{G(\rho_{1}) - G(\rho_{2})}{\Lq{1}} \leq
    &\frac{A}{C}\subnorm{\nabla V \cdot \nabla(\rho_{1} - \rho_{2}) + \nabla V(\rho_{1} - \rho_{2})}{\Lq{1}} \\
    + &\frac{B}{C}\subnorm{\nabla(W* \rho_{1})\cdot \nabla\rho_{1} - \nabla(W * \rho_{2})\cdot \nabla\rho_{2}}{\Lq{1}}\\
    + &\frac{B}{C}\subnorm{\Delta(W * \rho_{1})\rho_{1} - \Delta(W * \rho_{2})\rho_{2}}{\Lq{1}}.
  \end{aligned}
\end{equation}
We will continue our analysis term by term.  Starting with the first term and applying the Cauchy-Schwarz inequality we find that
\begin{equation}\label{eq:term1}
  \begin{aligned}
    \frac{A}{C}\subnorm{\nabla V \cdot \nabla(\rho_{1} - \rho_{2}) + \nabla V(\rho_{1} - \rho_{2})}{\Lq{1}}&\leq \frac{A}{C}\Big[ \subnorm{\nabla V}{\Lq{2}}\subnorm{\nabla(\rho_{1} - \rho_{2})}{\Lq{2}}+ \subnorm{\Delta V}{\Lq{2}}\subnorm{\rho_{1} - \rho_{2}}{\Lq{2}}\Big]. \\
  \end{aligned}
\end{equation}
Now, turning to the second term of \eqref{eq:expanded} and applying the Cauchy-Schwarz inequality we find:
\begin{equation}\label{eq:term2_intermediate}
  \begin{aligned}
    \frac{B}{C}\subnorm{\nabla(W* \rho_{1})\cdot \nabla\rho_{1} - \nabla(W * \rho_{2})\cdot \nabla\rho_{2}}{\Lq{1}} &\leq \frac{B}{C}\Big[\subnorm{\nabla W * \rho_{1}}{\Lq{2}} \subnorm{\nabla(\rho_{1} - \rho_{2})}{\Lq{2}}\\
    &+ \subnorm{\nabla W * (\rho_{1} - \rho_{2})\cdot \nabla \rho_{2}}{\Lq{1}}\Big ].
  \end{aligned}
\end{equation}
We will further expand the last term of \eqref{eq:term2_intermediate}.  Notice that:
\begin{equation}\label{eq:term2_intermediate2}
  \begin{aligned}
    \frac{B}{C}\subnorm{\nabla W * (\rho_{1} - \rho_{2})\cdot \nabla \rho_{2}}{\Lq{1}} \leq \frac{B}{C}\subnorm{\rho_{1} - \rho_{2}}{\Lq{2}}\sum_{i = 1}^{d}\subnorm{W_{x_{i}}}{\Lq{2}}\subnorm{\rho_{2, x_{i}}}{\Lq{1}}.
  \end{aligned}
\end{equation}
Now, combining \eqref{eq:term2_intermediate} and \eqref{eq:term2_intermediate2} we find
\begin{equation}\label{eq:term2}
  \begin{aligned}
    \frac{B}{C}\subnorm{\nabla(W* \rho_{1})\cdot \nabla\rho_{1} - \nabla(W * \rho_{2})\cdot \nabla\rho_{2}}{\Lq{1}} &\leq \frac{B}{C}\subnorm{\nabla W * \rho_{1}}{\Lq{2}} \subnorm{D(\rho_{1} - \rho_{2})}{\Lq{2}}\\
    &+ \frac{B}{C}\subnorm{\rho_{1} - \rho_{2}}{\Lq{2}}\sum_{i = 1}^{d}\subnorm{W_{x_{i}}}{\Lq{2}}\subnorm{\rho_{2, x_{i}}}{\Lq{1}}. \\
  \end{aligned}
\end{equation}
Finally, we turn to the third term of \eqref{eq:expanded}.  Notice that:
\begin{equation}\label{eq:term3}
  \begin{aligned}
    \frac{B}{C}\subnorm{\Delta(W * \rho_{1})\rho_{1} - \Delta(W * \rho_{2})\rho_{2}}{\Lq{1}} &= \frac{B}{C}\int_{\Omega}|\Delta(W * \rho_{1})\rho_{1} - \Delta(W * \rho_{2})\rho_{2}|dx\\
    &= \frac{B}{C}\int_{\Omega}|\Delta(W * \rho_{1})\rho_{1} - \Delta(W * \rho_{2})\rho_{1} + \Delta(W * \rho_{2})\rho_{1} - \Delta(W * \rho_{2})\rho_{2}| dx\\
    &= \frac{B}{C}\int_{\Omega}|\Delta(W * (\rho_{1} - \rho_{2})\rho_{1}) + \Delta(W * \rho_{2})(\rho_{1} - \rho_{2})|dx\\
    &\leq \frac{B}{C}\subnorm{\Delta W * \rho_{2}}{\Lq{2}}\subnorm{\rho_{1} - \rho_{2}}{\Lq{2}}\\
    &+\frac{B}{C} \subnorm{\Delta W}{\Lq{2}}\subnorm{\rho_{1} - \rho_{2}}{\Lq{2}}\subnorm{\rho_{1}}{\Lq{1}}.\\
  \end{aligned}
\end{equation}
Now, combining our results in \eqref{eq:term1}, \eqref{eq:term2} and \eqref{eq:term3} we find that:
\begin{equation}
  \begin{aligned}
    \subnorm{G(\rho_{1}) - G(\rho_{2})}{\Lq{1}} &\leq \frac{A}{C}\Big[ \subnorm{\nabla V}{\Lq{2}}\subnorm{\rho_{1} - \rho_{2}}{\Lq{2}}+ \subnorm{\Delta V}{\Lq{2}}\subnorm{\rho_{1} - \rho_{2}}{\Lq{2}}\Big]\\
    &+ \frac{B}{C}\Big[\subnorm{\nabla W * \rho_{1}}{\Lq{2}} \subnorm{\rho_{1} - \rho_{2}}{\Lq{2}}+ \subnorm{\rho_{1} - \rho_{2}}{\Lq{2}}\sum_{i = 1}^{n}\subnorm{W_{x_{i}}}{\Lq{2}}\subnorm{\rho_{2, x_{i}}}{\Lq{1}}\Big]\\
    &+ \frac{B}{C}\Big[\subnorm{\Delta W * \rho_{2}}{\Lq{2}}\subnorm{\rho_{1} - \rho_{2}}{\Lq{2}}+ \subnorm{\Delta W}{\Lq{2}}\subnorm{\rho_{1} - \rho_{2}}{\Lq{2}}\subnorm{\rho_{1}}{\Lq{1}}\Big]\\
    &:=\gamma(A,B,C)\subnorm{\rho_{1} - \rho_{2}}{\Lq{2}}
  \end{aligned}
\end{equation}
where $\gamma \geq 0$ and is independent of $\rho_{1}$ and $\rho_{2}$ as we have assumed a uniform bound on their norms (they are members of $H(\Omega, K)$). Therefore, if $W$ and $V$ are given, $\gamma$ can be made arbitrarily small if $C$ is sufficiently larger than $A$ and $B$.  Additionally we have that
\begin{align*}
  \subnorm{G(\rho_{1}) - G(\rho_{2})}{\Lq{2}} &\leq \beta \subnorm{G(\rho_{1}) - G(\rho_{2})}{\Lq{1}}
\end{align*}
for some $\beta > 0$ and therefore
\begin{equation}\label{eq:final_bound}
\begin{aligned}
    \subnorm{G(\rho_{1}) - G(\rho_{2})}{\Lq{2}} &\leq \beta \subnorm{G(\rho_{1}) - G(\rho_{2})}{\Lq{1}}\\
    &\leq \beta\gamma(A,B,C)\subnorm{\rho_{1} - \rho_{2}}{\Lq{2}}.
\end{aligned}
\end{equation}
Combining \eqref{eq:G_condition} and \eqref{eq:final_bound} we obtain:
\begin{align*}
    \Hnorm{\tilde{\rho_{1}} - \tilde{\rho_{2}}} \leq \alpha\subnorm{G(\rho_{1}) - G(\rho_{2})}{\Lq{2}}\leq\alpha\beta\gamma(A,B,C)\Hnorm{\rho_{1} - \rho_{2}}.
\end{align*}
Therefore for $C$ sufficiently larger than $A$ and $B$, the mapping $\tilde{A}$ is a contraction as desired.
\end{proof}

From a modeling point of view, Theorem \ref{thm:diffusion_stationary_states} shows that too strong a contribution from the diffusion term will destroy any of the qualitative behaviors we wish to capture.  We will see this explicitly in our numerical studies of \eqref{eq:macroscopic_model} in the next section.

\section{Numerics - Qualitative features of the macroscopic model}
We now present some numerical investigations of the macroscopic model derived in the previous section.

\begin{align}\label{equation}
\rho_{t} = A(\nabla \cdot (\nabla V\rho))  + B(\nabla \cdot((\nabla W * \rho)\rho)) + C\Delta\rho
\end{align}

These investigations are aimed at replicating qualitative properties of slime mold movement seen in \cite{nakagaki_t_obtaining_2004, nakagaki_maze-solving_2000, nakagaki_path_2001}, specifically aggregation around food sources while still maintaining a "connected" mass.  Recall that here, the function $V$ models the density of chemoattractants dispersed by food sources.  Therefore, $V$ will always take the form of a sum of radially symmetric positive functions where each term is centered on a food source.  We simulate \eqref{equation} via the "blob method for aggregation-diffusion equations" introduced in \cite{carrillo_blob_2019} which approximates \eqref{equation} by solving the ODE:

\begin{equation}\label{ODE}
  \begin{aligned}
    \dot{x}_{i}(t) = &-A(\nabla V(X_{i})) - B(\sum_{j=1}^{N}\nabla W(x_{i}(t) - x_{j}(t))m_{j})\\
    &-C\Bigg( \sum_{j=1}^{N}\Bigg(\Big(\sum_{k=1}^{N}\upvarphi_{\epsilon}(x_{j}-x_{k})m_{k}\Big)+\Big(\sum_{k=1}^{N}\upvarphi_{\epsilon}(x_{i} - x_{k})m_{k}\Big)^{-1} \Bigg)\nabla\upvarphi_{\epsilon}(x_{i} - x_{j})m_{j}\Bigg)
  \end{aligned}
\end{equation}

for a collection of $N$ particles, $(x_{1}(t),...,x_{N}(t), t\geq 0)$, who's initial positions are a regular grid on the domain on which we'd like to approximate \eqref{equation} and who's masses are given by $m_{i} = \rho(0, X_{i}(0))$.  To recover an approximation of $\rho$ from the positions of the particles we convolve the particle solution with a mollifier $\upvarphi_{\epsilon}$:
\begin{align*}
  \bar{\rho}(x,t) = \sum_{i}\upvarphi_{\epsilon}(x-X_{i}(t))m_{i}.
\end{align*}
We assume mollifiers are always of the form:
\begin{align*}
  \upvarphi_{\epsilon}(x) = \frac{1}{(4\pi\epsilon^{2})^{d/2}} e^{-|x|^{2}/4\epsilon^{2}},
\end{align*}
where $d$ is the dimension of the domain.  It is shown in \cite{carrillo_blob_2019} that under some regularity and growth conditions on $W$ and $V$ that $\bar{\rho}$ converges to the solution of \eqref{equation}.  As a first example we will consider \eqref{equation} in one dimension.  We will take the food source to be
\begin{align*}
  V(x) = -e^{-(1-x)^{2}} - e^{-(-1-x)^{2}},
\end{align*}
Therefore we can think of the food sources as being located at $x = -1$ and $x = 1$. For the parameters of the ODE we take:
\begin{itemize}
  \item $N = 100$ particles initially equally space on the interval $[-2.1,2.1]$, i.e. with separation $h = .042$
  \item $\epsilon = h^{.99}$ (Mollifier parameter)
  \item initial mass profile given by $\rho_{0}(x) = \frac{1}{\sqrt{2\pi\sigma^{2}}}e^-\frac{x^{2}}{2\sigma^{2}}$, $\sigma^{2} = 0.0625$ as in \cite{carrillo_blob_2019}.
\end{itemize}
Note that as the variance of the initial profile is very small that, at the particle level, slime mold particles have a very low probability of being found outside of $[-2.1,2.1]$ (initially).  In Figure \ref{fig:1D_fair} we perform a first simulation of \eqref{eq:macroscopic_model} in one dimension via the blob method in the so called "fair regime" ($A=B=C=1$).  Here, we choose a quadratic interaction kernel - specifically we impose that $W'(x) = x$.

\begin{figure}
  \centering
  \includegraphics[scale = .6]{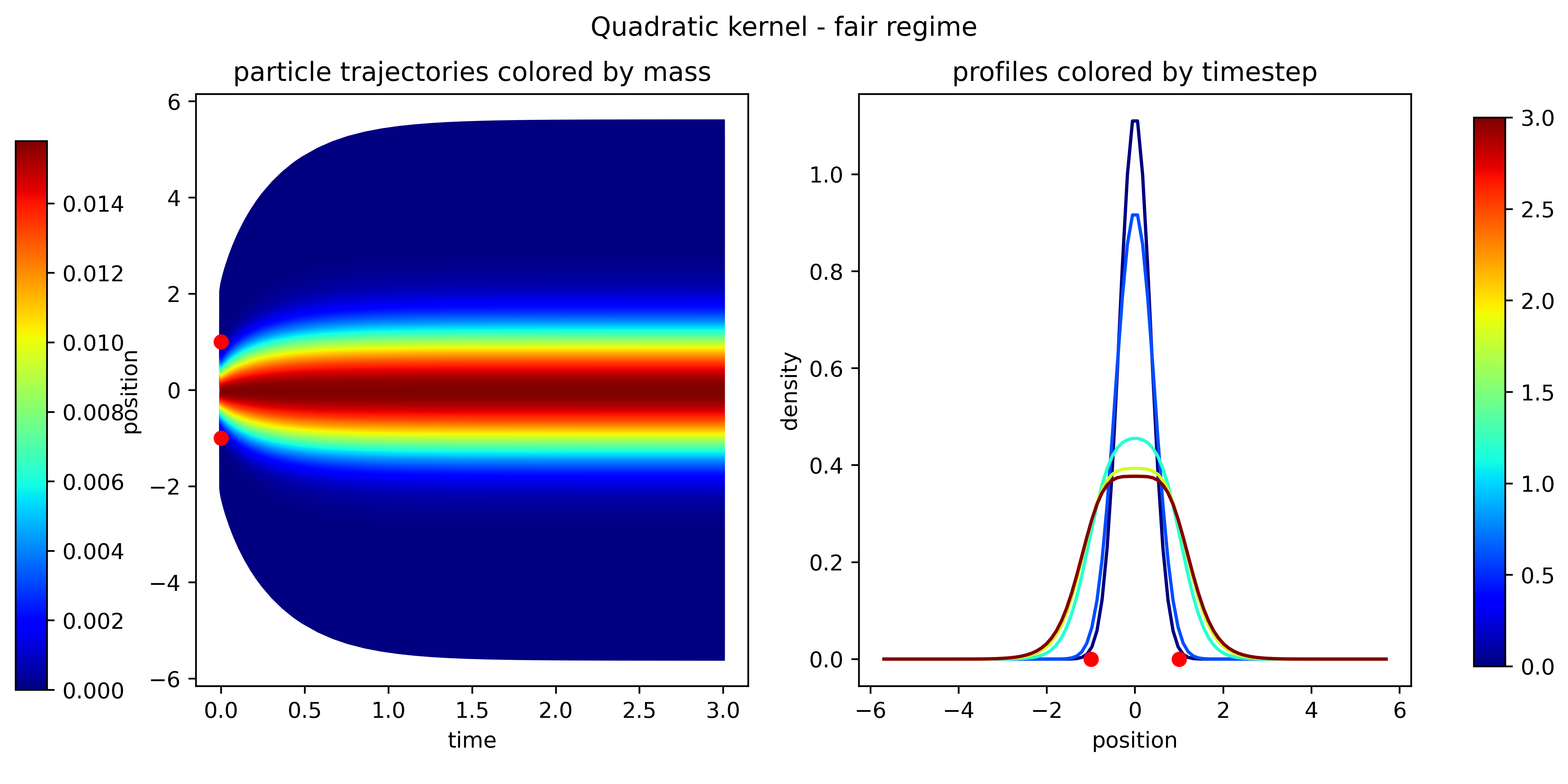}
  \caption{The evolution of the slime mold model \eqref{equation} in the fair regime ($A=B=C$).  Food source aggregation is not observed.}
  \label{fig:1D_fair}
\end{figure}

Qualitatively, we do not observe aggregation of slime mold mass around food sources to the degree present in \cite{nakagaki_t_obtaining_2004}. Therefore, motivated by this example, in the following we will attempt to modify \eqref{equation} in order to better model slime mold food seeking behavior. In Figure \ref{fig:1D_fair} we observe that diffusion appears to be the dominating effect.  Therefore, in order to examine the interplay of the diffusion term against the other two terms in \eqref{equation} we will employ the scaling parameters, $A,B,C >0$, in order to control the contributions from the drift term, interaction term and diffusion term respectively.  Roughly we examine two main regimes; "drift dominated" (where $A > B,C$) and "interaction dominated" (where $B > A, C$).  We also examine so called "competition regimes" where two scaling parameters are equal but dominate the third.  We do not examine the diffusion dominated regime as we've previously seen that if $A=B=C$ (we will refer to this as the "fair regime") then diffusion is the dominating effect.  In each regime we examine three different interaction kernels - quadratic, polynomial and Gaussian.  We choose the same ODE parameters as in the previous fair regime simulation.

\subsection{Food source dominated regime}
We will first examine simulations of \eqref{eq:macroscopic_model} in the regime where the drift term (which models attraction to the food sources) is the dominating factor.  We first examine the case where the kernel is quadratic and given by $W(x) = \frac{x^{2}}{2}$. Here, the kernel is purely attractive.  At the agent based level all agents exert a "pull" on other agents with a force that is proportional to their spacial seperation.  The results can be seen in Figure \ref{fig:drift_quadratic_1D}. Here, we qualitatively observe that in the case where $A=10$, $B=C=1$ that we have strong food source aggregation however, the mass profile splits into two distinct bumps - the strong contribution from the drift term prevents a connected mass from being maintained.  In the other three cases a connected mass is maintained and food source aggregation is observed.  However, in the two cases where $C=5$ the degree of aggregation is relatively weak likely due to the stronger contribution from the diffusion term.

\begin{figure}
    \centering
    \includegraphics[scale=.8]{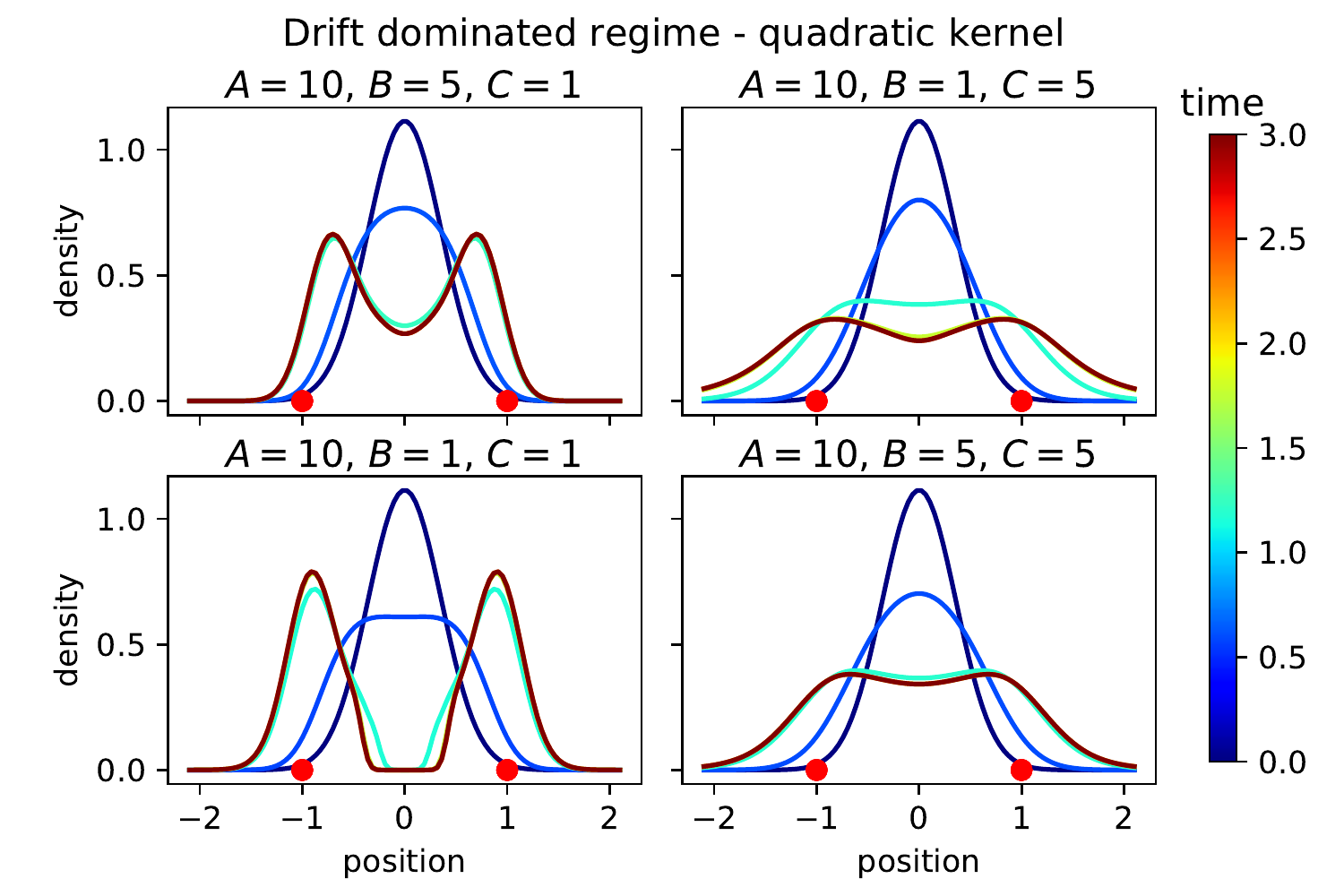}
    \caption{The evolution of \eqref{eq:macroscopic_model}  in the drift dominated regime with kernel given by $W(x)= \frac{x^{2}}{2}$.  Food source aggregation occured in two cases only one of which maintained a connected mass.}
    \label{fig:drift_quadratic_1D}
\end{figure}

%

Next, we examine the case where the kernel is polynomial and given by $W(x) = \frac{x^{4}}{4} - \frac{x^{2}}{2}$.  Intuitively, at the particle level, this means that this kernel is attractive at long-ranges but includes a short-range repulsion to prevent particles from becoming too close.  The results can be seen in Figure \ref{fig:drift_poly_1d}.  Here, we observe results that are qualitatively very similar to the results obtained in the quadratic case.  This is not surprising given that the polynomial kernel is largely attractive and at long ranges the degree of attraction is proportional to the separation between particles - analagous to the quadratic kernel.  However in the case where $A=10, B=5, C=1$ the evolution with a polynomial kernel does not maintain a connected mass in contrast to the evolution witha quadratic kernel.  This is likely due to the short-range repulsion that is present in the polynomial kernel.  At the particle level, particles in the middle of the domain are repelled away from eachother while simultaneously being pulled by a food-source - encouraging a separation.

\begin{figure}
    \centering
    \includegraphics[scale=.8]{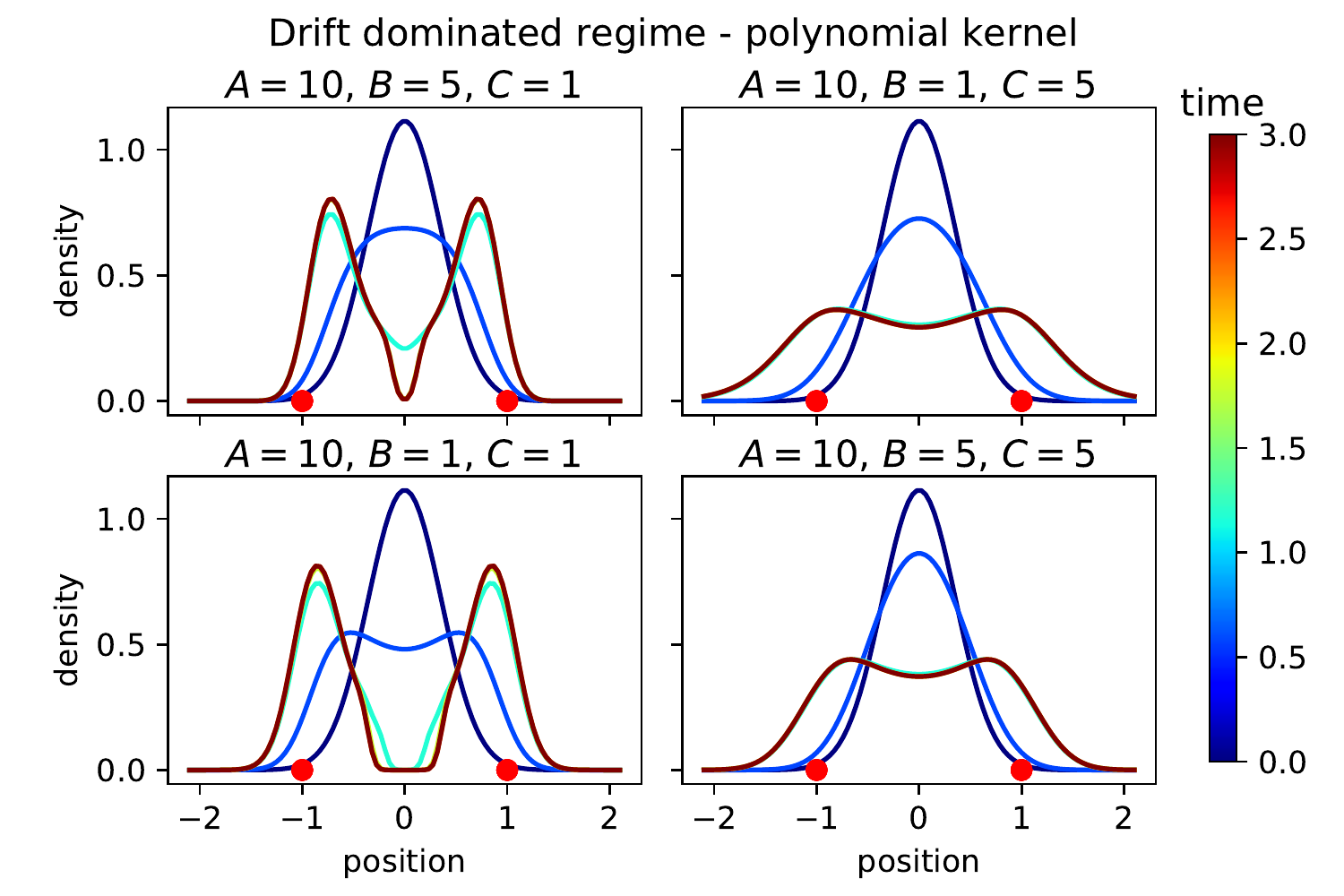}
    \caption{The evolution \eqref{eq:macroscopic_model}  in the drift dominated regime with kernel given by $W(x) = \frac{x^{4}}{4} - \frac{x^{2}}{2}$.  Strong food source aggregation is observed in two cases however neither maintains a connected mass.  Weak food source aggregation is observed in the other cases.}
    \label{fig:drift_poly_1d}
\end{figure}

Finally, we examine the case where the kernel is Gaussian and given by $W(x)=  \upvarphi_{.99}(x) = \frac{1}{(4\pi \cdot .99^{2})^{1/2}} e^{-x^{2}/4\cdot .99^{2}}$.  Here, similar to the quadratic kernel the interaction is purely attractive.  However, unlike the polynomial the strength of attraction between two particles is \textit{not} proportional to their separation.  Instead, attraction becomes stronger (but not in an unbounded fashion) as two particles become closer and weaker as they become farther apart.  This encodes the modeling assumption that interactions should be "local" - particles who are physically distant from eachother should not interact strongly.  The results can be seen in Figure \ref{fig:drift_gaus_1D}.

  Here, we observe food source aggregation in every case.  Similar to the polynomial and wuadratic kernels the aggregation is stronger when the contribution from the diffusion term is small.  However, similar to the polynomial kernel a connected mass is not maintained in either of these cases.  Additionally, in the cases where a connected mass is maintained we still observe a sharp drop off in the middle of the profile - indicating formation of a separation.  We beleive that the tendency to separate in this case is caused by a combination of the strong local interaction and the strong contribution from the food sources; particles that start close to a food source are "trapped" and remain close to the food source as the contribution from the food source and the local interaction reinforce eachother.  Then, particles towards the center of the domain feel a pull towards food sources form both the food source itself and the "trapped" particles.

\begin{figure}
    \centering
    \includegraphics[scale=.8]{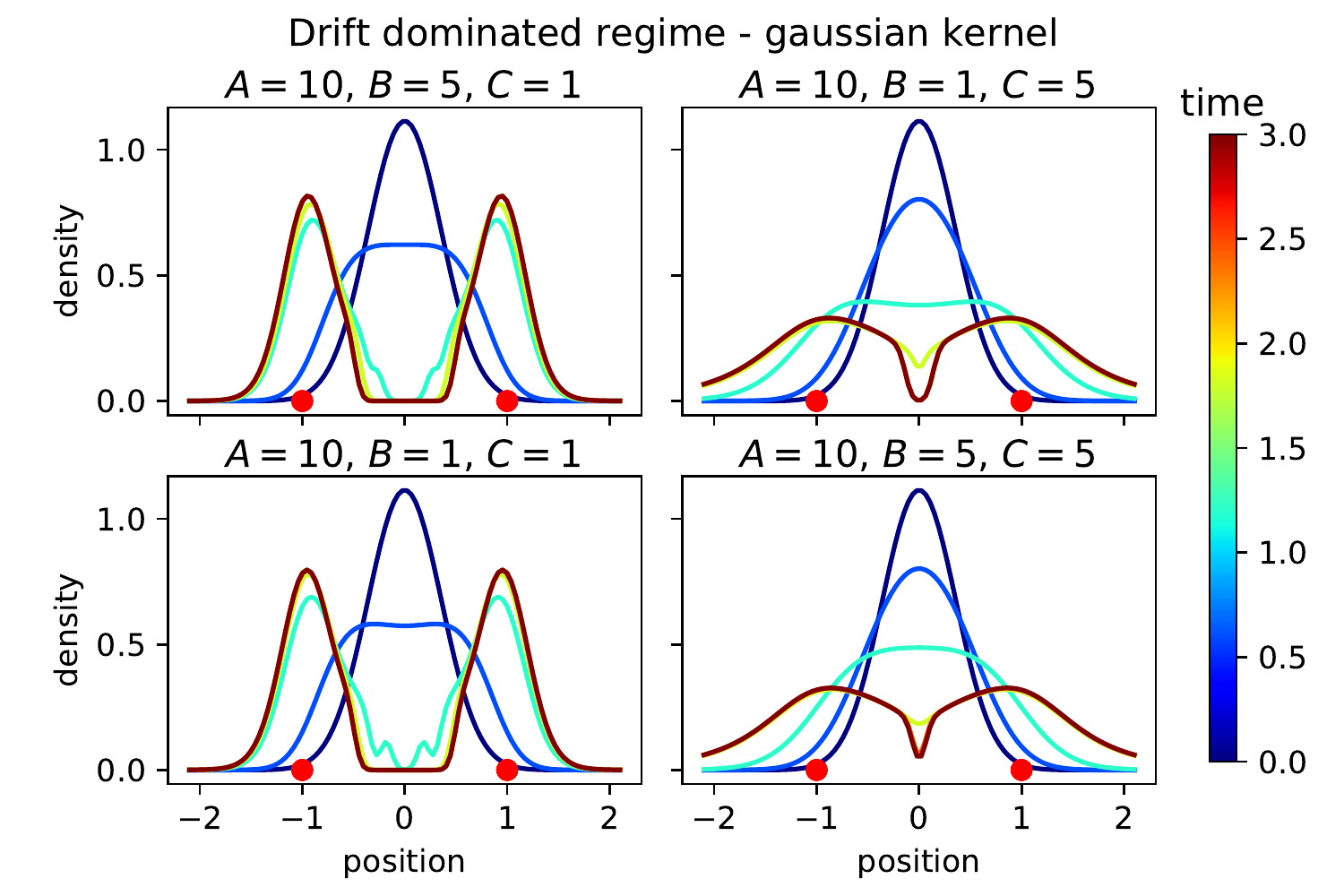}
    \caption{The evolution of \eqref{eq:macroscopic_model}  with kernel given by $W(x)=   \upvarphi_{.99}(x) = \frac{1}{(4\pi \cdot .99^{2})^{1/2}} e^{-|x|^{2}/4\cdot .99^{2}} $.  Food source aggregation occurs in all cases, however in two cases a connected mass is not maintained. }
    \label{fig:drift_gaus_1D}
\end{figure}

%

\subsection{Interaction dominated regime}
We now examine simulations of \eqref{eq:macroscopic_model} in the regime where the interaction term is the strongest contribution.  Again, we start with the case where the kernel function $W$ is quadratic and given by $W(x) = \frac{x^{2}}{2}$.  The results can be seen in Figure \ref{fig:interact_quadratic_1D}.  Here, we do not see any food source aggregation.  This can possibly be explained by the fact that in these simulations the initial profile is Gaussian. We have previously seen that Gaussians are stationary states for \eqref{eq:macroscopic_model} in the case of a quadratic kernel and no drift term.

\begin{figure}
    \centering
    \includegraphics[scale=.8]{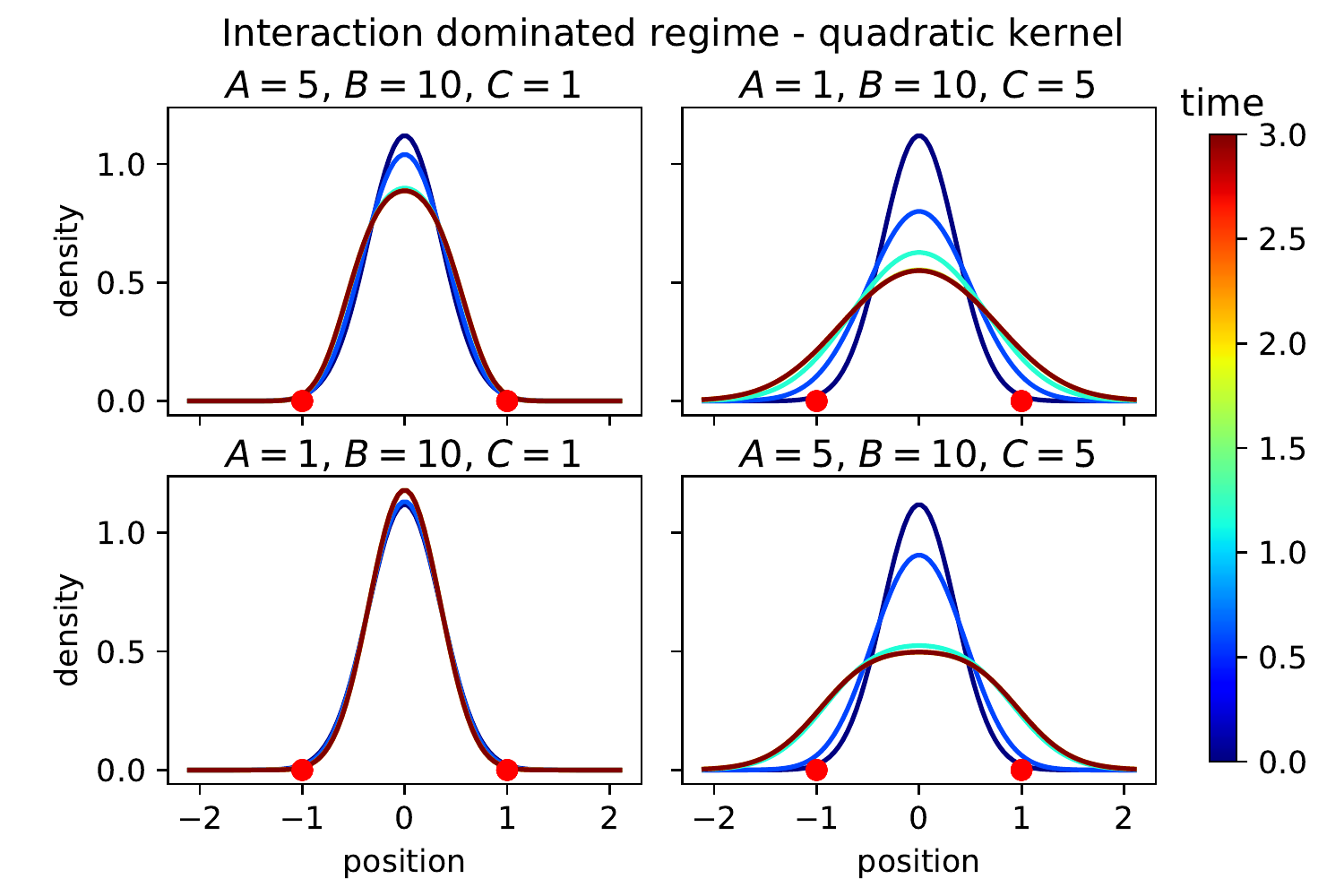}
    \caption{The evolution of \eqref{eq:macroscopic_model}  in the interaction dominated regime with kernel given by $W(x)= \frac{x^{2}}{2}$.  Food source aggregation is not observed.}
    \label{fig:interact_quadratic_1D}
\end{figure}

%
%
Next we examine the case where the kernel function is polynomial, of the same form as in the food source dominated regime.  The results can be seen in Figure \ref{fig:interact_poly_1d}.  Here, analagous to the food source dominated regime we see more pronounced aggregation in the cases where the diffusion term has the weakest contribution.  However the aggregation is less pronounced than the respective examples in the food source dominated regime due to the stronger contribution from the interaction term.  Unlike the analagous cases in the drift dominated regime we do not observe food source aggregation in the cases where $C = 5$, this is possibly due to the stronger contribution from the short range repulsion in the interaction term.

\begin{figure}
    \centering
    \includegraphics[scale=.8]{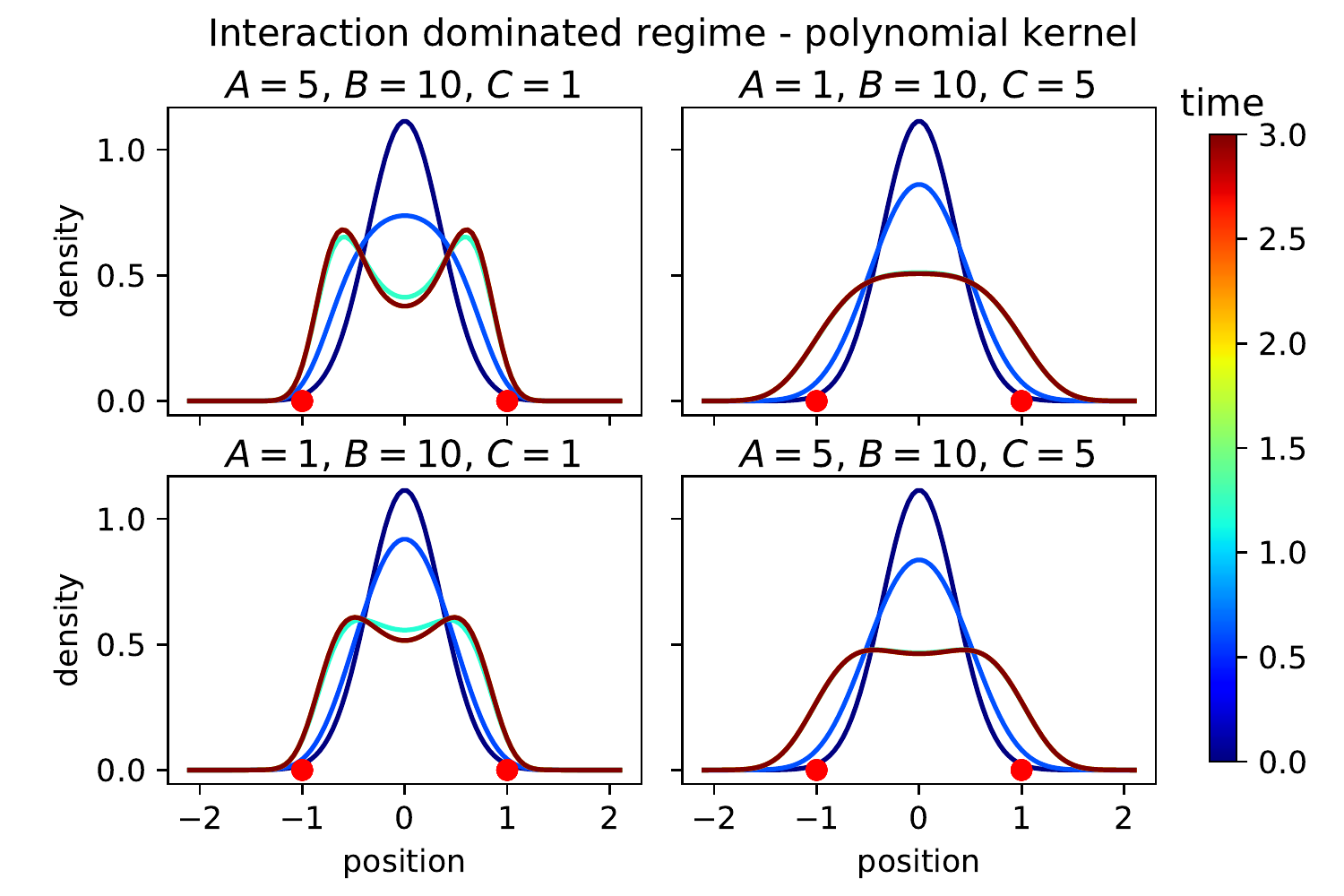}
    \caption{The evolution \eqref{eq:macroscopic_model}  in the interaction dominated regime with kernel given by $W(x) = \frac{x^{4}}{4} - \frac{x^{2}}{2}$.  Food source aggregation while maintaining a connected mass is observed in two cases.}
    \label{fig:interact_poly_1d}
\end{figure}

%
Finally, we examine the evolution in the case of a Gaussian kernel function.  The results can be seen in Figure \ref{fig:interact_gaus_1D}.  Here, in contrast to the respective cases in the drift dominated regime, we do not observe food source aggregation in the cases where the diffusion term has the second largest contribution.  Addtionally, unlike the analagous cases in the drift dominated regime and the polynomial kernel in the interaction dominated regime, we do not observe aggregation in the case where $B=10, C=1, A=1$.  We do observe aggregation in the case where $A=5, B=10, C=1$ however similar to the polynomial kernel the effect is less pronounced than in the analagous case in the drift dominated regime.  Both of these effects are consistent with a stronger contribution from the attractive kernel.
\begin{figure}
    \centering
    \includegraphics[scale=.8]{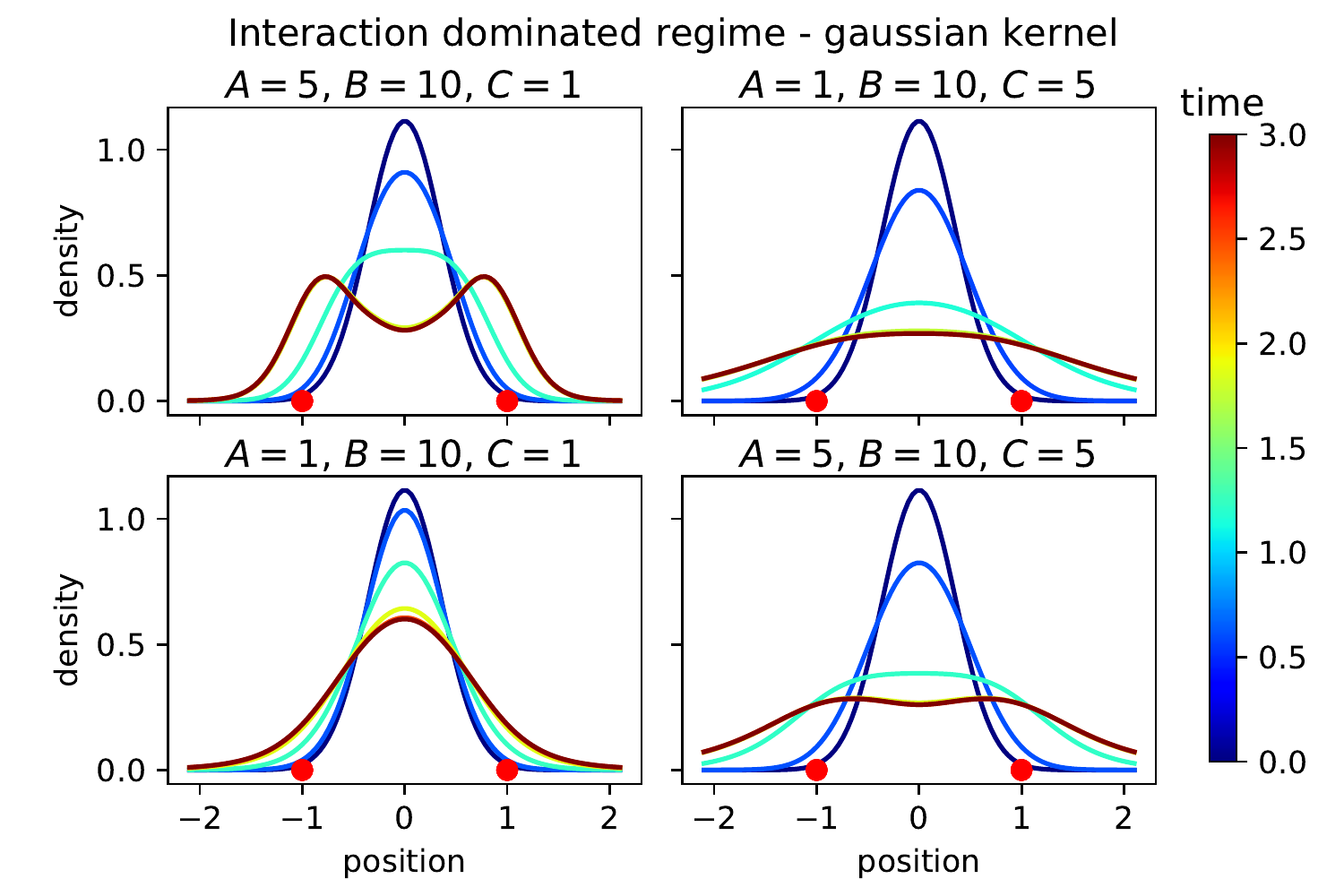}
    \caption{The evolution \eqref{eq:macroscopic_model}  in the interaction dominated regime with kernel given by $W(x) = \upvarphi_{.99}(x) = \frac{1}{(4\pi \cdot .99^{2})^{1/2}} e^{-|x|^{2}/4\cdot .99^{2}}$.  Food source aggregation while maintaining a connected mass is observed in one case.}
    \label{fig:interact_gaus_1D}
\end{figure}
%
%

\subsection{Competition regime}
We now simulate cases where two parameters dominate in competition. The results can be seen in Figure \ref{fig:quad_comp}, Figure \ref{fig:poly_comp}, and Figure \ref{fig:gaus_comp}.  Interestingly, across all three interaction kernels, pronounced food source aggregation is only observed in cases where $A=B=10$.  All of these cases except one exhibit the qualitative features we wish to capture as they do maintain a connected mass.  All other cases do not result in food source aggregation.  This can be explained by the fact that in all of these cases the diffusion term has a significant contribution.  The observation that strong contributions from the diffusion term dampen food source agregation is a general pattern that we have also observed in the other regimes and is additionally supported by the statement of Theorem \ref{thm:diffusion_stationary_states}.

  The cases in the competition regime where we do see food source aggregation also reinforce another pattern observed across all regimes; examples that have the qualitative features we wish to capture have strong contributions from the drift term, the interaction term, or both.  However, we also find that cases in which a connected mass was not maintained only occured where the drift term was dominant.  Likewise, cases in which food source aggregation did not occur in the absence of strong diffusion only occured in the interaction dominated regime.  These observations suggest a general modeling principal; contributions from the drift and diffusion terms should be scaled similarly and larger than the contribution from the diffusion term.

\begin{figure}
    \centering
    \includegraphics[scale = .6]{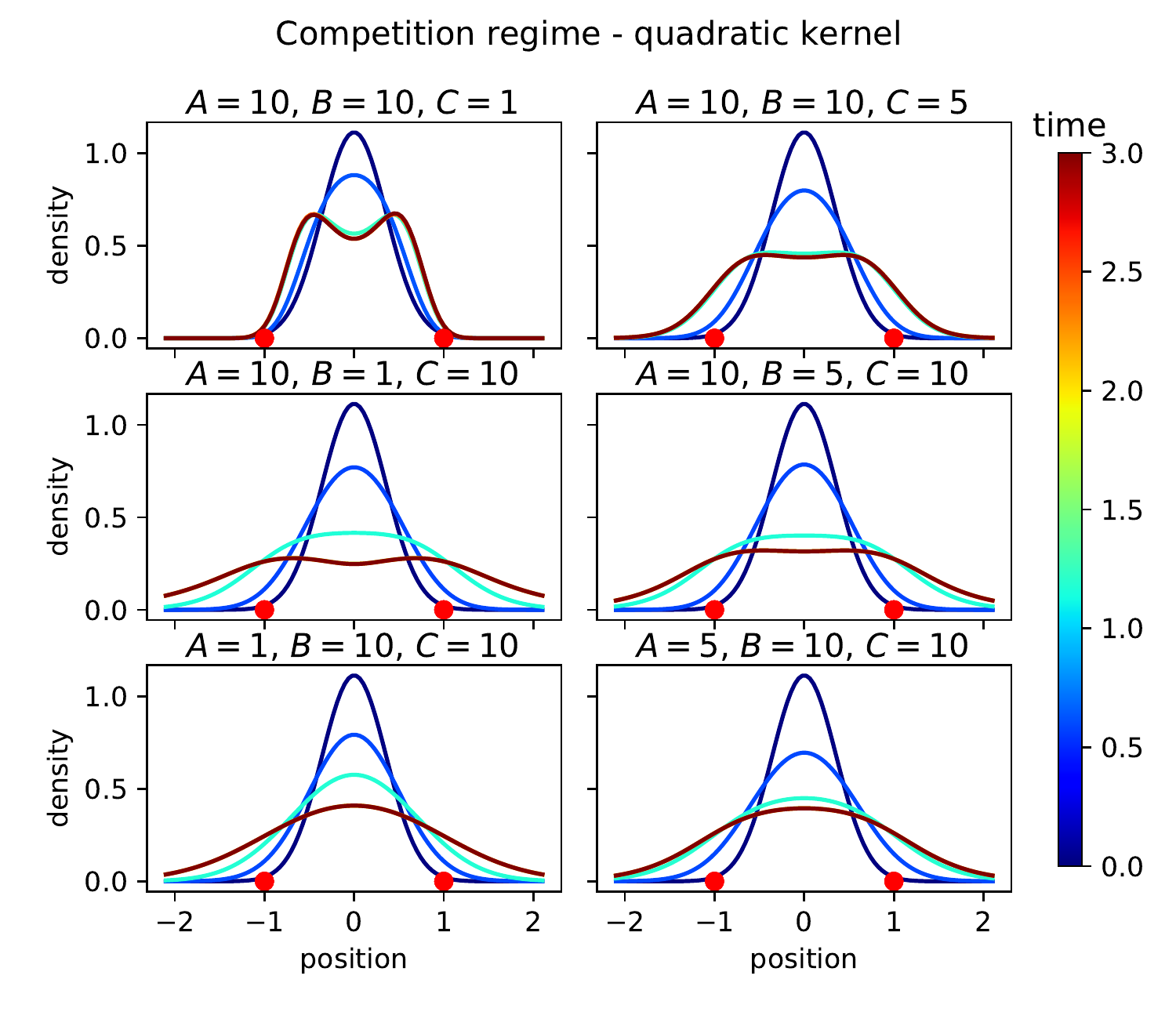}
    \caption{The evolution of \eqref{eq:macroscopic_model}  in the competition regime with kernel given by $W(x)= \frac{x^{2}}{2}$.  Food source aggregation is only observed in cases with a week diffusion contribution.}
    \label{fig:quad_comp}
\end{figure}

\begin{figure}
    \centering
    \includegraphics[scale = .6]{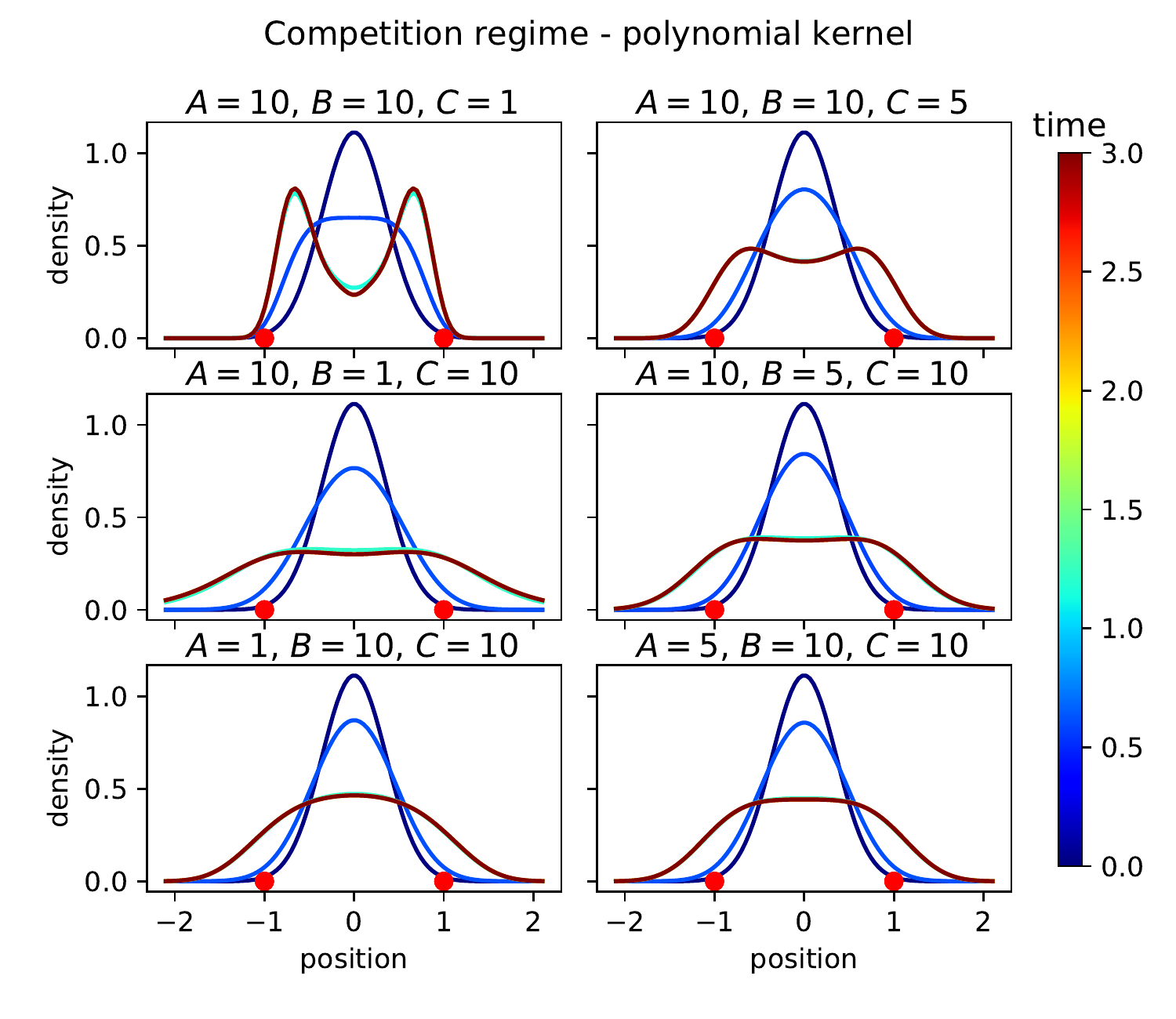}
    \caption{The evolution of \eqref{eq:macroscopic_model}  in the competition regime with kernel given by $W(x) = \frac{x^{4}}{4} - \frac{x^{2}}{2}$.  Food source aggregation is only observed in cases with a weak diffusion contribution.}
    \label{fig:poly_comp}
\end{figure}

\begin{figure}
    \centering
    \includegraphics[scale = .6]{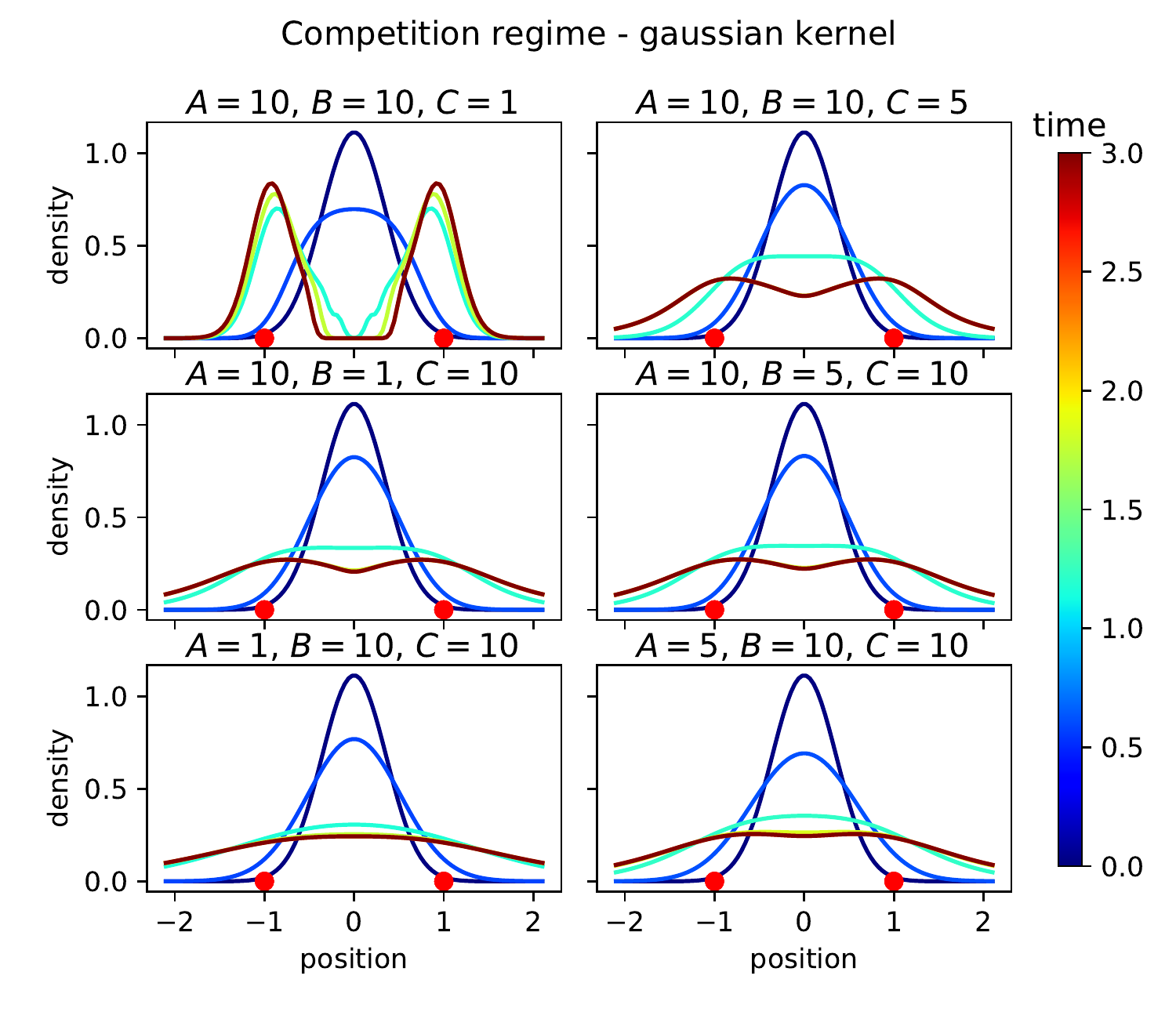}
    \caption{The evolution of \eqref{eq:macroscopic_model} in the competition regime with kernel given by $W(x) =  \frac{1}{(4\pi \cdot .99^{2})^{1/2}} e^{-|x|^{2}/4\cdot .99^{2}}$.  Food source aggregation is only observed in cases with a weak diffusion contribution.}
    \label{fig:gaus_comp}
\end{figure}



\section{Conclusion}
In this manuscript we presented a model hierarchy aimed at modeling the food seeking behavior of the slime mold Physarum Polycephalum.  We were principally interested in capturing the slime mold's ability to aggregate around disparate food sources while maintaining a connected mass.  We first presented a particle based model which includes three main features - a drift term to model a gradient of chemoattractant produced by food sources, an interaction term to model the slime mold's propensity to maintain a connected mass, and a diffusion term to model slime mold foraging behavior.  Additionally, we included scaling parameters for each of these three terms in order to study how varying their contributions could aid in modeling Physarum.  However simulating enough particles to realistically model the evolution of slime mold would be computationally intractable.  Therefore, under the assumption of propagation of chaos we showed that in the large particle limit the evolution of the particle model can be described by a macroscopic aggregation-diffusion equation which can be efficiently simulated.  Before embarking on simulating the equation we made some analytical observations via analysis of the equation's stationary states.  First, we found that in the case of no drift term that an assumption of Gaussian stationary states implies that the interaction kernel must be quadratic and discussed how this could be used as a heuristic to determine a kernel that realistically models physarum behavior.  Then, we showed that if the contribution from the diffusion term is sufficiently larger than the interaction and drift terms that the only possible stationary state is 0, deeming this regime as unsuitable for modeling Physarum.

We then ran a series of simulations of the macroscopic model in order to investigate how to scale the three contributing terms and choose the interaction kernel in order to model Physarum food seeking behavior.  We simulated three different parameter regimes; a regime in which the drift term dominated, a regime in which the interaction term dominated and a regime in which the dominating terms were in competition.  In each regime we examined three different interaction kernels, an attractive kernel with a local interaction, an attractive kernel with long range interaction, and a kernel that included a short range repulsion.  We observed cases that reproduced connected mass food source aggregation across all three regimes with all kernels however several general patterns did emerge.  We found that strong contributions from the diffusion term tend to dampen food source aggregation even in the presence of strong contributions from the other terms.  Additionally, we found that in the absence of strong diffusion, dominance of the drift term can result in mass separation while dominance of the interaction term can prevent food aggregation.  This suggests a general modeling strategy of scaling contributions from the drift and interaction terms similarly and larger than the diffusion term.

There are many ways in which this study could be extended.  Aggregation-diffusion equations have been widely studied in the case of no drift term \cite{chen_modeling_2018}.  A main finding is that there is a critical mass above which the interaction term becomes dominant and causes a "blow up" in the mass profile.  From the particle view this can be thought of intuitvely as all particles converging on one point.  Below the critical mass, the tendency of the diffusion term to cause the mass profile to spread prevents a blow up from occuring.  We did not observe this phenomena in any of our simulations (even in the interaction dominated regime).  It would be interesting to consider whether the presence of a drift term can prevent blow up from occuring at all or mitigate it in the sense that the critical mass at which blow up occurs becomes higher.  Additionally, as we did observe the desired qualitative behavior in several cases it would be useful to compare the evolution of our model to data from actual slime mold growth in a simple situation such as the two food source case we consider.  This would be an important first step in the quantitative validation of this model.

\section*{Acknowledgments}
Part of this work was carried out while D. Weber was visiting the University of Mannheim within the program IPID4all funded by the German Academic Exchange Service (DAAD). This work was supported by the DAAD project {\em Stochastic dynamics for complex networks and systems} (Project-ID 57444394).

\section*{Conflict of interest}
All authors declare no conflicts of interest in this paper.

\bibliography{bib}
\bibliographystyle{siam}

\end{document}